\documentclass[a4paper,11pt]{article}
\pdfoutput=1 

\usepackage[english]{babel}
\usepackage{textalpha}

\usepackage{jheppub} 
\usepackage{amsmath,amssymb,amsthm,cancel}
\usepackage{graphicx}
\usepackage{subcaption}
\usepackage[percent]{overpic}
\usepackage[dvipsnames]{xcolor}
\usepackage{multirow}
\usepackage{bbm}
\usepackage{todonotes}
\usepackage{tikz}
\usetikzlibrary{arrows,positioning,cd,calc,math,decorations.pathreplacing,decorations.markings,decorations.pathmorphing
}
\tikzset{wave/.style={decorate, decoration=snake}}
\def\nudge{.5}
\tikzset{axis/.style={ultra thick, Red!75!black, -latex, shorten <=-\nudge cm, shorten >=-2*\nudge cm}}
\tikzset{line/.style={ultra thick,black}}

\usepackage{epigraph}
\usepackage{xcolor}
\definecolor{MyBlue}{rgb}{0.25,0.5,0.75}
\colorlet{NextBlue}{MyBlue!20}
\colorlet{SecondBlue}{MyBlue!40}

\definecolor{purple}{rgb}{0.7,0,0.7}

\def\be{\begin{equation}}
\def\ee{\end{equation}}

\newcommand{\CI}{\mathcal{I}}
\newcommand{\CS}{\mathcal{S}}

\newcommand{\CQ}{\mathcal{Q}}

\newcommand{\CX}{\mathcal{X}}

\newcommand{\CC}{\mathcal{C}}

\renewcommand{\sl}{\mathfrak{sl}}

\newcommand{\IZ}{\mathbb{Z}}

\renewcommand{\Re}{{\rm Re\,}}
\renewcommand{\Im}{{\rm Im\,}}
\newcommand{\dd}{\mathrm{d}}

\newtheorem{theorem}{Theorem}

\newtheorem{remark}{Remark}

\title{
BPS Spectra and Algebraic Solutions of Discrete Integrable Systems
}

\author[a]{Fabrizio  Del Monte}

\affiliation[a]{School of Mathematics and Statistics, University of Sheffield, Hounsfield Road, Sheffield S3 7RH, United Kingdom}

\emailAdd{
f.delmonte.mp@gmail.com
}

\abstract{This paper extends the correspondence between discrete Cluster Integrable Systems and BPS spectra of five-dimensional $\mathcal{N}=1$ QFTs on $\mathbb{R}^4\times S^1$ by proving that algebraic solutions of the integrable systems are exact solutions for the system of TBA equations arising from the BPS spectral problem. This statement is exemplified in the case of M-theory compactifications on local del Pezzo Calabi-Yau threefolds, corresponding to q-Painlev\'e equations and $SU(2)$ gauge theories with matter. A degeneration scheme is introduced, allowing to obtain closed-form expression for the BPS spectrum also in systems without algebraic solutions. By studying the example of local del Pezzo 3, it is shown that when the region in moduli space associated to an algebraic solution is a ``wall of marginal stability'', the BPS spectrum contains states of arbitrarily high spin, and corresponds to a 5d uplift of a four-dimensional nonlagrangian theory.}

\begin{document}
\maketitle
\setlength{\epigraphwidth}{.5\textwidth}
\epigraph{Nature uses only the longest threads to weave her patterns, so that each small piece of her fabric reveals the organization of the entire tapestry.}{Richard P. Feynman}

\section{Introduction}

In recent years, there has been a new wave of interest in the study of the BPS spectrum of five-dimensional Quantum Field Theories with eight supercharges, after a decade of progress with the four-dimensional case. The strongest motivation comes from observing that compactification of type IIA string theory on a local Calabi-Yau threefold doesn't simply produce a standard four-dimensional $\mathcal{N}=2$ theory, but rather a five-dimensional one \cite{Seiberg:1996bd,Morrison:1996xf,Douglas:1996xp,Intriligator:1997pq}: the hidden presence of the M-theory circle leads to an infinite number of fields in the four-dimensional theory, which are Kaluza-Klein (KK) modes on the five-dimensional circle, so that the BPS spectrum of such theories contains highly nontrivial nonperturbative information about M-theory itself\footnote{In this paper we study only the BPS \textit{particle} spectrum, although in general, five-dimensional theories also have more exotic BPS string states.}. Prototypical examples of local Calabi-Yau threefolds appearing in this context are total spaces of the canonical bundles over a complex surfaces $S$, where $S$ is either a $\mathbb{P}^1\times\mathbb{P}^1$ or a (pseudo-) del Pezzo surface $dP_n$, these latters being blowups of $\mathbb{P}^2$ at $n$ (possibly nongeneric) points. Apart from the case of local $\mathbb{P}^2$, the 5d theories resulting from these geometries admit low-energy $SU(2)$ gauge theory phases with matter.

The advantage of a string-theoretic mindset towards these CFTs is that stable BPS states are realized geometrically by D0, D2, D4-branes wrapping holomorphic cycles inside a resolution $X$
of the (typically singular) Calabi-Yau geometry, so that the lattice of BPS charges is the even cohomology lattice with compact support
\begin{equation}\label{eq:GammaIntro}
\Gamma=H^0_{cpt}(X)\oplus H^2_{cpt}(X)\oplus H^4_{cpt}(X).
\end{equation}
BPS states are then mathematically described as objects in the derived category $D^b(X)$ of coherent sheaves on $X$.
By this correspondence, exact computations of BPS spectra for five-dimensional theories have nontrivial counterpart in the Donaldson-Thomas theory of the corresponding geometry \cite{Bridgeland:2016nqw}. 

Our main tool will be the so-called \textit{BPS quiver} of the theory, a term introduced for four-dimensional theories in \cite{Cecotti:2011rv} and generalized to the present context in \cite{Closset:2019juk}, appearing also in related physics literature under the name of fractional brane quiver \cite{Douglas:1996sw,Hanany:2005ve,Franco:2005rj,Feng2005,Yamazaki:2008bt}.
The determination of the BPS spectrum for five-dimensional theories on a circle is typically much more involved than for four-dimensional ones, and until recently very little was known. In \cite{Bonelli2020}, inspired by the recently discovered relation between partition functions of five-dimensional gauge theories (or equivalently, Topological String partition functions \cite{Eguchi:2003sj,Taki:2007dh}) and q-Painlev\'e tau functions \cite{Bershtein2016,Bonelli2017,Bershtein2017,Jimbo:2017,Matsuhira2018,Bonelli2020}, a new strategy for the computation of the BPS spectrum was proposed. The general idea behind this approach is to introduce a discrete integrable system from the underlying Calabi-Yau geometry, which encodes hidden quantum symmetries of the five-dimensional QFT. When applied to an appropriate set of elementary states, the discrete time evolution of the integrable system generates the BPS spectrum. Different examples were considered, and a conjectural expression advanced for the case of local $\mathbb{P}^1\times\mathbb{P}^1$ in \cite{Closset:2019juk} was readily reproduced with this new method, that outlined a clear path forward for more general cases.

The relation between Cluster Integrable Systems and BPS spectra was further clarified in the subsequent paper \cite{DelMonte2021}, where it was shown how to associate to the Cluster Integrable System special chambers in the moduli space of the theory  (technically, in the moduli space of stability conditions of the $CY_3$), dubbed collimation chambers, in which the BPS spectrum can be computed exactly. These chambers have the property that the spectrum is comprized of infinite towers of states, with central charges all limiting to the positive real axis. The name comes from the analogy between rays of light and rays on the complex plane of central charges of the BPS states, as we might think of the BPS spectrum as a light beam, with rays all travelling in the same direction confined in a given width, as depicted in Figure \ref{Fig:P1P1coll}.

\begin{figure}[h]
\begin{center}
\begin{tikzpicture}[scale=1]
  \foreach \x in {-6.5,...,6.5}{
       \draw[->, line width = .5mm] (0,0)--(\x,1);
       \draw[->, line width = .5mm] (0,0)--(\x,-1);}
       \draw[->,color=red,line width = .5mm] (0,0)--(1,0);
       \draw[->,color=red,line width = .5mm] (0,0)--(5,0);
       \draw[->,color=red,line width = .5mm] (0,0)--(-1,0);
       \draw[->,color=red,line width = .5mm] (0,0)--(-5,0);
       \node at (5.5,0) {\Large\dots};
       \node at (-5.5,0) {\Large\dots};
       \node at (7,1) {\Large\dots};
       \node at (-7,1) {\Large\dots};
       \node at (7,-1) {\Large\dots};
       \node at (-7,-1) {\Large\dots};
\end{tikzpicture}
\end{center}
\caption{Schematic picture of central charges in the collimation chamber in local $\mathbb{P}^1\times\mathbb{P}^1$. The black arrows correspond to towers of hypermultiplets, while the red arrows to vector multiplets and corresponding Kaluza-Klein modes. Towers of hypermultiplets such as the one in the figure are also called Peacock patterns.}
\label{Fig:P1P1coll}
\end{figure}
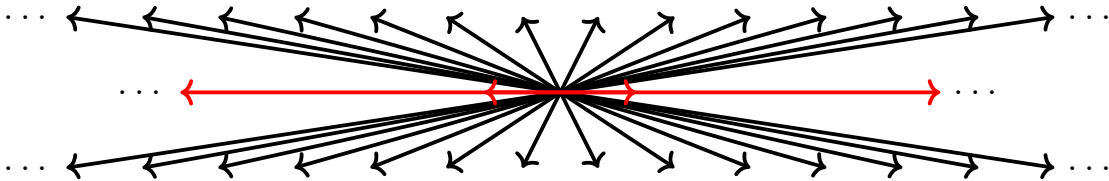

Another aspect of the relation between BPS spectra and Cluster Integrable Systems comes from the fact that the String theory corrections to the central charges are described by a set of TBA equations that appeared for the first time in the work of Gaiotto, Moore and Neitzke on four-dimensional $\mathcal{N}=2$ QFTs \cite{Gaiotto:2008cd,Gaiotto:2014bza}, and were then shown to describe the D-instanton corrections to the central charges in type IIB String Theory \cite{Alexandrov:2013yva,Alexandrov:2021prq,Alexandrov:2021wxu}:
\be\label{eq:TBA-conformal}
\begin{split}
	\log \CX_\gamma(\epsilon) 
	=
	\frac{Z_\gamma }{\epsilon} 
	-\frac{\epsilon}{\pi i}\sum_{\gamma'>0}\Omega(\gamma')\langle\gamma,\gamma'\rangle
	\int_{\ell_{\gamma'}}
	\frac{d \epsilon'}{(\epsilon')^2 - (\epsilon)^2}
	\log(1-\sigma(\gamma') \CX_{\gamma'}(\epsilon')) \,.
	\\
\end{split}
\ee
Here $\gamma\in \Gamma$ is a vector in the lattice of BPS charges \eqref{eq:GammaIntro} and $Z_\gamma$ is the central charge of the corresponding BPS state. $\Omega(\gamma)$ is the BPS degeneracy, coinciding with the DT invariant of the coherent sheaf $\gamma$, and $\sigma(\gamma)=-1$ if $\Omega(\gamma)=1$, while $\sigma(\gamma)=1$ if $\Omega(\gamma)=-2$ (we will not need other cases here). The $\CX_\gamma$ can be also regarded as solutions to the Bridgeland Riemann-Hilbert problem, which endows the moduli space of stability conditions of the  $CY_3$ with a geometric structure known as Joyce structure \cite{Bridgeland2019}. In \cite{DelMonte2022a}, it was shown that the TBA equation \eqref{eq:TBA-conformal} in the collimation chamber of local $\mathbb{P}^1\times\mathbb{P}^1$ can be rephrased as the q-Painlev\'e equation of symmetry type $A_1^{(1)}$. Furthermore, it was observed that within the collimation chambers there can existed a \textit{fine-tuned stratum} where the solution to \eqref{eq:TBA-conformal} receives no $\epsilon$-corrections, and corresponds to the \textit{algebraic solution} of the corresponding q-Painlev\'e equation, so that it was conjectured that such exact solutions should correspond to algebraic solutions of the Cluster Integarble System.

\textbf{Contents and Results:}

\begin{table}
\centering
\begin{tabular}{ |c|c|c|c|c| } 
 \hline
 Local $CY_3$ & Quiver & Collimation Chamber & BPS spectrum & Alg. Sol. \\
 \hline
$dP_5$ & Figure \ref{Fig:QuiverdP5} & Eq. \eqref{eq:CC1defdP5} & Eq. \eqref{eq:SpectrumdP5}* & Eq. \eqref{eq:dP5algsol}  \\ \hline
$dP_4$ & Figure \ref{Fig:QuiverdP4} & Eq. \eqref{eq:CC1defdP4} & Eq. \eqref{eq:SpectrumdP4}  & No \\ \hline
$dP_3$ & Figure \ref{Fig:QuiverdP3} & Eq. \eqref{eq:CC1defdP3} & Eq. \eqref{eq:SpectrumdP3}*  & Eq. \eqref{eq:exsolP3} \\ 
		&		& Eq. \eqref{eq:CC1defdP3p} & Eq. \eqref{eq:SpectrumdP3p}, \eqref{eq:SpectrumdP3real} & No \\ \hline
$dP_2$ & Figure \ref{Fig:QuiverdP2} & Eq. \eqref{eq:CC1defdP2} & Eq. \eqref{eq:SpectrumdP2}  & No \\ \hline
$dP_1$ & Figure \ref{Fig:QuiverdP1} & Eq. \eqref{eq:CC1defdP1} & Eq. \eqref{eq:SpectrumP1P1}  & No \\ \hline
$\mathbb{P}^1\times\mathbb{P}^1$ & Figure \ref{Fig:QuiverP1P1} & Eq. \eqref{eq:CC1defP1P1} $^\dagger$ & Eq. \eqref{eq:SpectrumP1P1}* & Eq. \eqref{eq:dP5algsol} \\
\hline
 \multicolumn{5}{| c |}{$*$: from \cite{DelMonte2021} \qquad $\dagger:$ from \cite{DelMonte2022a}}\\
\hline
\end{tabular}
\caption{Main Results of the paper}\label{Table:MainRes}
\end{table}

The main results of this paper consist of Theorem \ref{thm:ExSo}, together with the results summarized in Table \ref{Table:MainRes}. 

After briefly recalling how BPS quivers arise in five-dimensional SCFTs, in Section \ref{sec:TBAalg} we generalize and make precise the identification of exact solutions to the TBA equations \eqref{eq:TBA-conformal} with algebraic solutions of q-Painlev\'e equations, and more generally cluster integrable systems. To this end we prove Theorem \ref{thm:ExSo}, stating a precise set of conditions under which such an exact solution exists. These conditions are equivalent to the invariance under certain \textit{folding} transformation of the corresponding q-Painlev\'e equation \cite{Bershtein:2021gkq}, which is a property characterizing their algebraic solutions, so that the conjecture of \cite{DelMonte2022a} is effectively proven.

The theorem is then applied in Section \ref{sec:degens}, where the exact solutions are written explicitly for the cases of local $dP_5$ and local $dP_3$, realizing five-dimensional $SU(2)$ Super Yang-Mills with four and two fundamental hypermultiplets respectively. In the case of $dP_5$ and $dP_3$, starting from the algebraic solutions it is possible to obtain an infinite number of closed-form solutions to the TBA equations, which are rational solutions of q-Painlev\'e VI and III respectively. These are obtained by applying B\"acklund transformations to the algebraic solution, physically corresponding to appropriate sequences of dualities of the 5d theory, and while they can be written down explicitly, display all-order $\epsilon$ corrections, in contrast to the algebraic solutions.

The problems encountered in \cite{DelMonte2021} in finding collimation chambers for local $dP_4,\, dP_2,\, dP_1$ are explained by their lack of symmetry with respect to the other cases, signaled by the absence of folding transformation in the corresponding q-Painlev\'e equation. Nonetheless, in Section \ref{sec:degens} we implement a degeneration procedure that produces new collimation chambers, with explicit BPS spectrum and stability conditions, for these missing cases, completing the picture of collimation chambers and BPS spectra for five-dimensional $SU(2)$ super Yang-Mills up to $N_f=4$. The degeneration procedure amounts geometrically to the blow-down of exceptional 2-cycles in the local del Pezzo geometries, or equivalently to the holomorphic decoupling of hypermultiplets in the corresponding gauge theory. At the level of integrable systems it is the confluence of q-Painlev\'e equations, as described by Sakai's classification by symmetry type \cite{Sakai2001} in Figure \ref{Fig:Sakai}.

In Section \ref{sec:dP3AD} we show what happens when the stability condition associated to an algebraic solution lies on a wall of marginal stability.
By perturbing away from the algebraic solution, it is possible to find a stability condition which would still correspond to a collimation chamber, since it yields infinite towers of hypermultiplets which accumulate on the positive real axis. However, instead of finding mutually local vector multiplets,
on the real axis there is a (likely infinite) number of mutually non-local higher spin states, so that we are still lying on a wall of marginal stability.
The structure of the quiver suggests that a further deformation of this chamber might be related to a 5d uplift of a 4d Argyres-Douglas theory, as was pointed out in \cite{Bonelli2020}.

Finally, in Section \ref{sec:Concl} we make several concluding remarks about possible generalizations of this work: these include a yet unexplained observation on the connection between the BPS spectrum of 5d pure $SU(2)$ Super Yang-Mills (local $\mathbb{P}^1\times\mathbb{P}^1$) and the 4d $\mathcal{N}=2^*$ theory (2-Markov quiver), and extension to higher-rank gauge theories and five-dimensional uplifts of $E_n$ Minahan-Nemeschansky theories.

\begin{figure}[t]
\begin{center}
\begin{tikzpicture}[scale=.9]
\node at ($(0,0)$) {$E_8^{(1)}$};
\node at ($(2,0)$) {$E_7^{(1)}$};
\node at ($(4,0)$) {$E_6^{(1)}$};
\node at ($(6,0)$) {$E_5^{(1)}$}; 		
\node at ($(8,0)$) {$E_4^{(1)}$};
\node at ($(10,0)$) {$E_3^{(1)}$};
\node at ($(12,0)$) {$E_2^{(1)}$};
\node at ($(14,0)$) {$A_1^{(1)}$};
\node at ($(16,0)$) {$A_0^{(1)}$};
\node at ($(14,2)$) {$\begin{matrix}A_1^{(1)} \\|\alpha|^2=8\end{matrix}$};

\node at (14,0.7) {$\underline{\mathbb{P}^1\times\mathbb{P}^1}$};
\node at (14,3) {$\underline{dP_1}$};
\node at (12,0.7) {$\underline{dP_2}$};
\node at (10,0.7) {$\underline{dP_3}$};
\node at (8,0.7) {$\underline{dP_4}$};
\node at (6,0.7) {$\underline{dP_5}$};

\draw[thick,->](0.5,0) to (1.5,0);
\draw[thick,->](2.5,0) to (3.5,0);
\draw[thick,->](4.5,0) to (5.5,0);
\draw[thick,->](6.5,0) to (7.5,0); 		
\draw[thick,->](8.5,0) to (9.5,0);
\draw[thick,->](10.5,0) to (11.5,0);
\draw[thick,->](12.5,0) to (13.5,0);

\draw[thick,->](12.4,0.5) to (13.2,1.3);
\draw[thick,->](14.7,1.3) to (15.5,0.5);

\draw[thick,->](6.3,-0.5) to (7.1,-1.8);
\draw[thick,->](8.3,-0.5) to (9.1,-1.7);		
\draw[thick,->](10.3,-0.5) to (11.1,-1.8);
\draw[thick,->](12.4,-0.5) to (13.1,-1.5);
\draw[thick,->](14.3,-0.5) to (15.1,-1.8);

\node at ($(7.5,-2)$) {$D_4^{(1)}$};
\node at ($(9.5,-2)$) {$D_3^{(1)}$};		
\node at ($(11.5,-2)$) {$D_2^{(1)}$};
\node at ($(13.5,-2)$) {$\begin{matrix}A_1^{(1)} \\|\alpha|^2=4\end{matrix}$};
\node at ($(15.5,-2)$) {$A_0^{(1)}$};

\draw[thick,->](8,-2) to (9,-2);
\draw[thick,->](10,-2) to (11,-2);				
\draw[thick,->](12,-2) to (12.7,-2);
\draw[thick,->](14.3,-2) to (15,-2);

\draw[thick,->](9.8,-2.5) to (10.6,-3.8);				
\draw[thick,->](11.8,-2.5) to (12.6,-3.8);
\draw[thick,->](14,-2.7) to (14.6,-3.8);

\node at ($(11,-4)$) {$A_2^{(1)}$};
\node at ($(13,-4)$) {$A_1^{(1)}$};				
\node at ($(15,-4 )$) {$A_0^{(1)}$};

\draw[thick,->](11.5,-4) to (12.5,-4);				
\draw[thick,->](13.5,-4) to (14.5,-4);

\draw[thick,->] (10,-0.5) to (10.8,-3.5);
\draw[thick,->] (14.5,-0.2) to[out=315,in=45] (16.2,-2.5)  to[out=210,in=30] (13.5,-3.7) ;
\draw[thick,->] (16,-0.5) to[out=290,in=90] (17,-2) to[out=270,in=45] (15.5,-3.7);
\draw[thick,->](14.7,1.3) to (15.3,-1.5);
\draw[thick,->] (15.5,-2.5) to(15.3,-3.5);

\end{tikzpicture}
\end{center}
\caption{Sakai's Classification of discrete Painlev\'e equations by symmetry type. The affine root lattice coincides with the flavour sublattice of the corresponding geometry, above which we wrote down the corresponding geometry.}\label{Fig:Sakai}
\end{figure}
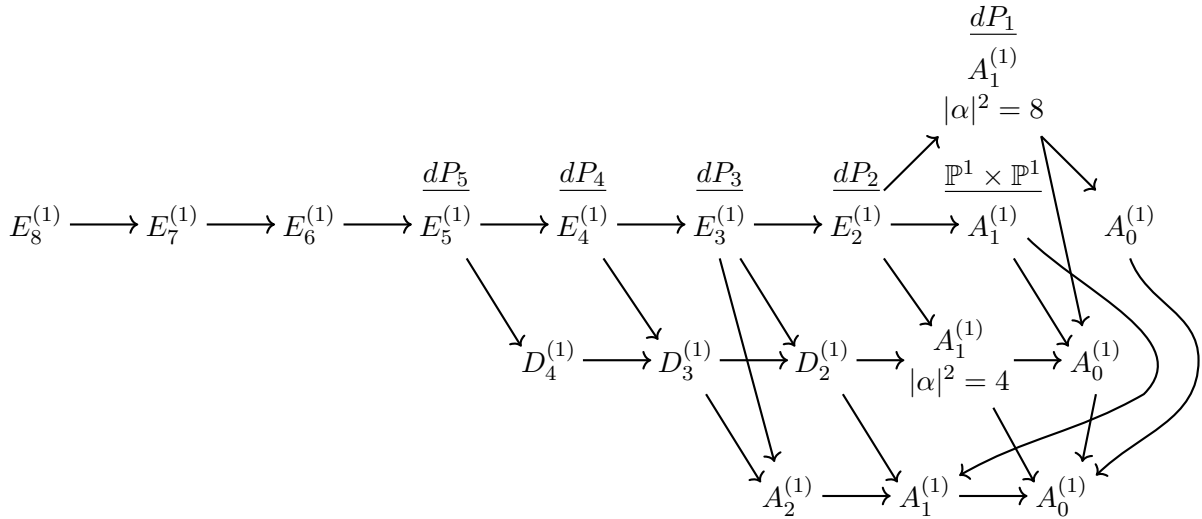
\section*{Acknowledgements:}
The author would like to thank M. Bershtein, T. Bridgeland, A. Grassi, K. Ito, N. Joshi, P. Longhi, B. Pioline, P. Roffelsen for invaluable discussions at various stages of this work. 

\section{TBA equations and algebraic solutions of Cluster Integrable Systems}\label{sec:TBAalg}

Let us start with some basic terminology. By a quiver $\CQ$, it is meant an oriented graph consisting of nodes connected by arrows. Here we always consider quivers with no arrows from a node to itself (loops), nor pairs of arrows connecting two nodes in opposite directions (2-cycles), and we label the nodes by numbers $1,\dots,|\CQ|$, where $|\CQ|$ denotes the size of the quiver, i.e. the number of its nodes. The quiver is then encoded in its (antisymmetrized) adjacency matrix $B$, whose entries $B_{ij}$ are equal to the number of arrows from node $i$ to node $j$, with the convention that outgoing arrows carry positive sign, while incoming arrows carry negive sign. A representation of a quiver $\CQ$ is an assignment of a vector space $V_i$ to each node $i$ of $\CQ$, and a linear map $\phi_{ij}$ for each arrow $i\rightarrow j$. In the context of quiver representation theory, to every node $i$ of the quiver one can associate an angle $\theta_i$, and the collection $\{\theta_i\}_{i=1}^{|\CQ|}$ is called a $\theta$-stability condition \cite{King1994MODULIOR}\footnote{In the present context of D-branes on Calabi-Yau threefolds, this is also related to Douglas $\Pi$-stability \cite{Douglas:2000ah}, mathematically formalized by the notion of Bridgeland stability \cite{bridgeland2007stability}.}, or simply a stability condition since there will not be any chance of confusion in this paper. The space of values that the $\theta_i$ can assume is called its moduli space of stability conditions.

In this paper, all the quivers are \textit{BPS quivers}, arising in the following way: BPS particles of M-theory on a Calabi-Yau threefold $X$ are described by M2- or M5-branes wrapping compact even-dimensional cycles in $X$, which are described by D0, D2, D4-branes in type IIA String Theory. They are labeled by their BPS charges, which are vectors of Chern characters of compactly supported coherent sheaves
\begin{equation}
\gamma\in\Gamma:=H_{cpt}^{even}(X,\mathbb{Z})=H^0_{cpt}(X,\mathbb{Z})\oplus H^2_{cpt}(X,\mathbb{Z})\oplus H^4_{cpt}(X,\mathbb{Z}).
\end{equation}
The low-energy dynamics of these particles is described by an $\mathcal{N}=4$ supersymmetric quantum mechanics associated to the quiver \cite{Douglas:1996sw,Denef:2002ru,Alim:2011ae}, typically determined by dimer model/brane tiling techniques (see \cite{Kennaway:2007tq,Yamazaki:2008bt,Franco:2017jeo} for comprehensive reviews of brane tilings) and the generators $\gamma_i$ of $\Gamma$ associated to the nodes of the quiver are the so-called fractional branes of the Calabi-Yau, which are hypermultiplet states. The adjacency matrix of the quiver is the antisymmetric Euler pairing in the basis $\gamma_i$ of fractional branes, which is identified with the physical Dirac pairing of the corresponding BPS states. For the local $CY_3$'s that we consider, which are total spaces of canonical bundles over an algebraic surface $S$, this is
\begin{align}
B_{ij}=\langle\gamma_i,\,\gamma_j\rangle:=\chi\left(\mathcal{E}_{\gamma_i},\,\mathcal{E}_{\gamma_j} \right)-\chi\left(\mathcal{E}_{\gamma_i},\,\mathcal{E}_{\gamma_j} \right), && \chi\left(\mathcal{E}_\gamma,\mathcal{E}_{\gamma'} \right)=\int_S\mathrm{ch}(\mathcal{E}_{\gamma}^\lor)\mathrm{ch}(\mathcal{E}_\gamma)\mathrm{Td}(S),
\end{align}
where $\mathrm{ch}$ and $\mathrm{Td}$ denote respectively the Chern and Todd class, $\mathcal{E}_\gamma$ is a sheaf with Chern vector gamma and $\mathcal{E}_{\gamma}^\lor$ is its dual. We say that two states are \textit{mutually local} if they have vanishing pairing. When the central charges of two mutually nonlocal states are aligned, they are called marginally stable, and such pathological regions of the moduli space of stability conditions are called walls of marginal stability.

The central charge, a linear function $Z:\Gamma\rightarrow\mathbb{C}$, describes the mass of BPS states through the BPS bound $M_\gamma=|Z_\gamma|$. Its phase $\theta_\gamma$ give the Fayet-Ilioupoulos couplings of the quantum mechanical model, and the collection $\theta_{\gamma_i}$ is called a stability condition \footnote{The precise relation between central charge phases, various notions of stability, and Fayet Ilioupoulos couplings requires to introduce some more terminology from quiver representation theory that will not be used in the rest of the paper, see Section 2 of \cite{DelMonte2021} for more details and references.}. The stable representations of the quiver are the stable BPS particles of the QFT, and they can be depicted by ray diagrams, where each BPS state is represented through the corresponding vector in the complex plane of central charges. Because of this, we will refer with slight abuse of terminoloy to the $Z_{\gamma_i}$ as stability data. For the five-dimensional theories considered in this paper, there is always a preferred direction in the plane of central charges, the real axis, since there always is a D0-brane state (skyscraper sheaf) representing Kaluza-Klein modes on the 5d circle with central charge
\begin{equation}
Z_{D0}=\frac{2\pi}{R},
\end{equation}
where $R$ is the radius of the five-dimensional circle. In such theories, one usually expects the real axis to constitute an accumulation ray in the plane of central charges, since the Kaluza-Klein towers of states coming from dimensional reduction along $S^1$ are realized in String Theory by towers of D0-branes with central charge $Z_{n\,D0}=2\pi n/R$. Besides the accumulation ray along the real axis, it is a general feature of supersymmetric theories that ``higher'' spin multiplets, which in this context means any state which is not a hypermultiplet, are either limiting rays of an infinite sequence of hypermultiplet states, or they are contained within a cone in the complex plane of central charges, whose boundaries are limiting rays. In fact, one can classify the possible types of chambers in the moduli space of the theory based on how the central charges of BPS states are organized:

\begin{itemize}
\item The simplest chambers are \textit{finite}, with a finite number of stable states which are hypermultiples. This is not possible in 5d due to the KK modes, so our chambers will always be infinite.
\item If there is only one accumulation ray on which lies the central charge of a vector multiplet, the chamber is called \textit{tame}
\item  More generally, one will have various accumulation rays, and between them there will be a cone where it is expected to find particles of arbitrary high spin organized in Regge trajectories \cite{Galakhov:2013oja,Cordova:2015vma}. Such a chamber is called a \textit{wild} chamber.
\item In the five-dimensional setting the real axis is an accumulation ray in the complex plane of central charges due to the presence of KK modes. This means that in order to have a tame chamber it will be necessary for all the vector multiplet states to have real central charges, and to avoid walls of marginal stability they must be mutually local. We can relax this condition, and allow the presence of higher spin states, as long as they also lie on the real axis and are mutually local: a chamber with these properties was named \textit{collimation chamber} in \cite{DelMonte2021}, as all the accumulation rays are collimated on the real axis. Although the notion of collimation chamber and that of tame chamber are in principle distinct, all the known examples of collimation chambers are also tame.
\end{itemize}

\subsection*{The mutation method}

Given a quiver and a stability condition, one can obtain (at least part of) the BPS spectrum by using the mutation method \cite{Alim:2011ae,Alim:2011kw}. Let us briefly recap the main ideas behind this procedure. The nodes of a BPS quiver correspond to an integral basis of simple (hypermultiplet) objects in the charge lattice, i.e. they are indecomposable and any other state can be written in terms of them as a linear combination with positive coefficients. Furthermore, one can choose an appropriate half-plane, referred to as the positive half-plane, where lie all the central charges of the quiver nodes. This amounts to a choice of what we call particle and antiparticle, since the central charge of an antiparticle is the opposite of the central charge of the corresponding particle.

If we start to rotate clockwise the choice of half-plane, at some point the ray of a BPS state for some node of the quiver will exit the positive half-plane, inducing a change in the quiver description corresponding to a mutation of the BPS quiver at the corresponding node \cite{Berenstein:2002fi}. If we mutate at the node $k$, we will have a new basis of the charge lattice, and the antisymmetric pairing must change accordingly, as:
\begin{align}\label{eq:Bmutation}
	\mu_k(B_{ij})= \begin{cases}
		-B_{ij}, & i=k\text{ or }j=k, \\
		B_{ij}+\frac{B_{ik}|B_{kj}|+B_{kj}|B_{ik}|}{2}, & \text{ otherwise}.
	\end{cases}
\end{align}
\begin{equation}\label{eq:gammamutation}
\mu_k(\gamma_j)= \begin{cases}
-\gamma_j, & j=k, \\
\gamma_j+[B_{jk}]_+\, \gamma_k, & \text{otherwise}.
\end{cases}
\end{equation}
This is just a change of basis and the new quiver just corresponds to a dual description of the same physics, so that the charges of the new nodes of the quiver must have been also in the original spectrum. By iterating this procedure, one produces stable hypermultiplet states in a given chamber, with higher spin multiplets appearing as limiting vectors of infinite sequences of mutations. In finite chambers, the  mutation  method exhausts the whole BPS spectrum. For collimation chambers, it is possible to obtain in this way all the hypermultiplet states, and then use additional permutation symmetries of the quiver to determine the states on the limiting ray \cite{DelMonte2021}. 

\subsection*{Exact solutions of TBA equations}

Once the BPS spectrum in a particular chamber is known, it is possible to formulate the so-called \textit{BPS Riemann-Hilbert problem} \cite{Bridgeland2019} associated to a $CY_3$ as the problem of solving the system \eqref{eq:TBA-conformal} of TBA equations \cite{Alexandrov:2021wxu}. We will now prove that under certain assumptions, such problem admits a classically exact solution, i.e. one with no $\epsilon$-corrections. 

\begin{theorem}\label{thm:ExSo}
Under the following assumptions:
\begin{enumerate}
\item There exist a permutation symmetry $\pi$ of the quiver and of the BPS spectrum, acting nontrivially on the charges $\gamma_i$, such that $\pi^N=id$ for some $N>1$.
\item $\pi$ is such that $\{\gamma_i+\pi(\gamma_i)+\dots\pi^{N-1}(\gamma_i)\}_{i=1}^{|\CQ|}\subseteq\Gamma_f$;
\end{enumerate}
Then the stability condition 
\begin{align}\label{eq:algstabthm}
Z_{\gamma_i}=Z_{\pi(\gamma_i)}
\end{align}
yields a semiclassically exact solution of the TBA equations \eqref{eq:TBA-conformal}:
\begin{align}\label{eq:TBAExSol}
\CX_{\gamma_i}=\CX_{\pi(\gamma_i)}=e^{\frac{Z_{\gamma_i}}{\epsilon}}.
\end{align}
\end{theorem}
\begin{proof}
Let $\CS$ be the BPS spectrum in the chamber associated to the stability condition \eqref{eq:algstabthm}. By assumption 1 it is invariant under $\pi$, so that
\begin{equation}
\CS=\CS_0\sqcup\pi(\CS_0)\sqcup\dots\sqcup\pi^{N-1}(\CS_0), \qquad \CS_0\subset\CS.
\end{equation}
The TBA equations can then be written as
\begin{equation}
\begin{split}
\log \CX_\gamma & =\frac{Z_\gamma}{\epsilon}-\frac{\epsilon}{i\pi}\sum_{j=0}^{N-1}\sum_{\gamma'\in\pi^j(\CS_0)}\Omega(\gamma')\langle\gamma,\gamma'\rangle\CI_{\gamma'}  \\
& = \frac{Z_\gamma}{\epsilon}-\frac{\epsilon}{i\pi}\sum_{\gamma'\in\CS_0}\sum_{j=0}^{N-1}\Omega(\pi^j(\gamma'))\langle\gamma,\pi^j(\gamma')\rangle\CI_{\pi^j(\gamma')},
\end{split}
\end{equation}
where we defined
\begin{equation}
\CI_{\gamma'}:=\int_{\ell_{\gamma'}}
	\frac{d \epsilon'}{(\epsilon')^2 - (\epsilon)^2}
	\log(1-\sigma(\gamma') \CX_{\gamma'}(\epsilon')).
\end{equation}
Since $\pi$ is simply a permutation of the nodes leaving the quiver invariant, for any $\gamma\in\Gamma$ we have $\Omega(\gamma)=\Omega(\pi(\gamma))$, and $\sigma(\gamma)=\sigma(\pi(\gamma))$, and the integration contour in the TBA equation also satisfies $\ell_{\gamma_i}=\ell_{\pi(\gamma_i)}$.  

When solving the TBAs order by order in $\epsilon$, the variables $\CX_\gamma$ in the integral of the order $\epsilon^{n+1}$ are the solution of the TBAs at order $\epsilon^{n-1}$. The first correction will be 
\begin{equation}
\begin{split}
\log \CX_\gamma \simeq\frac{Z_\gamma}{\epsilon}-\frac{\epsilon}{i\pi}\sum_{\gamma'\in\CS_0}&\Omega(\gamma')\langle\gamma,\gamma'+\pi(\gamma')+\dots+\pi^{N-1}(\gamma')\rangle\\
&\times\int_{\ell_{\gamma'}}\frac{d\epsilon'}{(\epsilon')^2-(\epsilon)^2}\log\left(1-\sigma(\gamma')e^{Z_{\gamma'}/\epsilon'}\right) ,
\end{split}
\end{equation}

On the other hand, since $\gamma'+\pi(\gamma')+\dots+\pi^{N-1}(\gamma')\in\Gamma_f$, we have
\begin{equation}
\langle\gamma,\gamma'+\pi(\gamma')+\dots+\pi^{N-1}(\gamma')\rangle=0,
\end{equation}
so that
\begin{equation}
\log\CX_{\gamma}=\frac{Z_\gamma}{\epsilon}+O(\epsilon^3).
\end{equation}
The first correction vanishes, but the same argument applies order by order, so that in fact the corrections will vanish to all orders. Then
\begin{equation}
\CX_\gamma=\exp\left(\frac{Z_\gamma}{\epsilon} \right)
\end{equation}
is an exact solution to the system of TBAs.
\end{proof}

\begin{remark}
For the solution \eqref{eq:TBAExSol} to be physically meaningful, we need  that $\langle\gamma_i,\pi(\gamma_i)\rangle=0$, since we are aligning the corresponding central charges. If the pairing is nonzero, it will still be true that \eqref{eq:TBAExSol} provides an exact solution to the TBA equations, but we will be lying on a wall of marginal stability. We will see in Section \ref{sec:dP3AD} that even in this latter case, it is possible to deform away from the fine-tuned stability condition in a controlled way, and discover interesting new phases of the 5d theory.
\end{remark}
\begin{remark}
This theorem can also be directly rephrased in a purely algebro-geometric setting without reference to the TBA equations \footnote{Many thanks go to T. Bridgeland, who brought this fact to the author's attention.}. In broad terms, it would state that given a BPS structure in the sense of \cite{Bridgeland:2016nqw} with symmetry $\pi$, the stability condition \eqref{eq:algstabthm} yields the solution \eqref{eq:TBAExSol} to the BPS Riemann-Hilbert Problem. This is because under these assumptions the jump condition
\begin{equation}\label{eq:WallCrossing}
\CX_\gamma\mapsto\CX_\gamma\prod_{Z_{\gamma'}\in\ell_\gamma}\left(1+\sigma(\gamma)\CX_\gamma \right)^{\langle\gamma,\gamma'\rangle\Omega(\gamma)}
\end{equation}
is trivially solved by $\CX_\gamma=e^{\frac{Z_\gamma}{z}}$.
\end{remark}

\section{Algebraic solutions, decoupling and exact BPS spectra of local del Pezzos}\label{sec:degens}

We now briefly introduce $\CX$-cluster variables, that can be used to construct (nonautonomous) cluster integrable systems associated to toric $CY_3$ \cite{Bershtein2017,Bershtein:2018srt,Mizuno2020}. We will show that the notion of algebraic solutions of the Cluster Integrable System coincides with the semiclassically exact solution of the TBA in Theorem \ref{thm:ExSo}. Indeed, as was shown explicitly for local $\mathbb{P}^1\times\mathbb{P}^1$ in \cite{DelMonte2022a}, the TBA equations \eqref{eq:TBA-conformal} in a collimation chamber are equivalent to the discrete equation of the integrable system, so that the general solution of the TBAs will correspond to a transcendental solution, rather than an algebraic one (we will briefly review how the argument applies for local $dP_5$ at the end of section \ref{sec:dP5}). A possible intermediate behaviour is given by rational solutions, that we will construct explicitly.

To each node $i$ of the quiver $\CQ$ we associate a $\mathbb{C}^\times$-valued $\CX$-cluster variable $\CX_i:=\CX_{\gamma_i}$, so that for a generic $\gamma=\sum_{i}n_i\gamma_i\in\Gamma$ we have
\begin{equation}
\CX_\gamma:=\prod_{i}\CX_i^{n_i}.
\end{equation}
The $\CX$-cluster variety is constructed by patching local charts of $\CX$-cluster variables by birational transformations called mutations. Under a mutation $\mu_k$ at a node $k$ of the quiver, the quiver and BPS charges transform according to \eqref{eq:Bmutation}, \eqref{eq:gammamutation}, while the $\CX$-cluster variables transform as
\begin{align}\label{eq:Xmutation}
		\mu_k(\CX_j)= \begin{cases}
		\CX_j^{-1}, & j=k, \\
		\CX_j\left(1+\CX_k^{\text{sgn} B_{jk}}\right)^{B_{jk}} , & \text{ otherwise}.
	\end{cases}
\end{align}
Other transformations that have to be considered are permutations of the nodes of the quiver, and the inversion
\begin{align}
\iota:\quad B_{ij}\rightarrow- B_{ij},\quad \CX_i\rightarrow \CX_i^{-1}.
\end{align}
The adjacency matrix of the quiver defines a Poisson bracket on the $\CX$-cluster variety, for which the coordinates $\CX_i$ are log-canonical:
\begin{equation}\label{eq:ClusterPoisson}
\{\CX_i,\CX_j\}=B_{ij}\CX_i\CX_j.
\end{equation}
The Casimirs of this algebra are $\CX$-cluster variables associated to elements $\gamma\in\ker B:=\Gamma_f\subset\Gamma$, the flavour sublattice of $\Gamma$. For local del Pezzos, such sublattices are isomorphic to affine root lattices, according to the diagram \ref{Fig:Sakai}, so that we can write them in terms of cluster variables associated to the affine roots $\alpha_i$ (see Appendix \ref{App:WeylGroups}),
\begin{align}\label{eq:AffineRoots}
	a_i:=\CX_{\alpha_i}, && q:=\CX_{D0} ,
\end{align}
where the D0 brane charge is the null root of the affine root lattice. These are also called multiplicative root variables. In the following, we will consider cluster variables as solutions to the TBA equations \eqref{eq:TBA-conformal}: in terms of central charges/stability data, the multiplicative root variables are simply
\begin{equation}
a_i=e^{\frac{Z_{\alpha_i}}{\epsilon}},\qquad q=e^{\frac{Z_{D0}}{\epsilon}}=e^{\frac{2\pi}{R\epsilon}},
\end{equation}
since central charges corresponding to flavour charges do not receive $\epsilon$-corrections in the TBA equations \eqref{eq:TBA-conformal}.

Sequences of mutations, permutations and inversion that preserve the quiver (but that can act nontrivially on the $\CX$-cluster variables) are Poisson maps on the cluster variety. The group of such transformation is called the cluster modular group, and in our present case it contains (often being isomorphic to) the extended affine Weyl group of the corresponding algebra \cite{Bershtein2017,Mizuno2020}. Among such transformation, a special role is performed by affine translations, that are discrete time evolutions for a corresponding (nonautonomous) cluster integrable system, which in this case are given by q-Painlev\'e equations \footnote{These can be seen as a deformation of cluster integrable systems arising from dimer models on the torus \cite{Goncharov2013,Fock2014}. The undeformed cluster integrable system describes the Seiberg-Witten geometry of the 5d gauge theory \cite{Brini:2009nbd,Eager:2011dp}, while the deformation takes into account the stringy ($\epsilon$-) corrections, similarly as what happens in four dimensions when going from a Hitchin System to Painlev\'e equations \cite{Bonelli2016}.}

It is a discrete flow among the leaves of \eqref{eq:ClusterPoisson}: an affine translation will act on the simple roots as
\begin{equation}
\alpha_i\mapsto\alpha_i+n_i\gamma_{D0},
\end{equation}
so that the Casimirs get transformed according to 
\begin{equation}
a_i\mapsto q^{n_i}a_i.
\end{equation}
Once an affine translation is chosen as a discrete time evolution, the remaining ones act as symmetries of the system, its so-called B\"acklund transformations.
The importance of these flows for us will be two-fold: on the one hand, it was shown in \cite{Bonelli2020,DelMonte2021} that discrete time evolutions of q-Painlev\'e equations, when acting on the charges through the mutation rule \eqref{eq:gammamutation},
generate all the hypermultiplet BPS states outside the accumulation ray starting from the initial generators $\gamma_i$ associated to the nodes of the quiver. When acting on the $\CX$-cluster variables, they instead yield the q-difference equations of the cluster integrable system.
The remaining flows, i.e. the B\"acklund transformations, generate an infinite number of solutions to the TBA equations \eqref{eq:TBA-conformal} starting from the simple solution of Theorem \ref{thm:ExSo}.


\subsection{Local $dP_5$ and the algebraic solution to q-Painlev\'e VI}\label{sec:dP5}
We consider the BPS quiver associated to local $dP_5$\footnote{\label{fn:PseudodP}The quiver \ref{Fig:QuiverdP5} is also associated to ``pseudo local $dP_5$'', which is a toric geometry first introduced in \cite{Feng:2002fv}, constructed by blowing up $\mathbb{P}^2$ at non-generic points. The distinction between the two geometries enters in the superpotential associated to the quiver and in additional bidirectional arrows in the quiver, both of which however do not play a role in our considerations, see \cite{Hanany:2012hi,Beaujard:2020sgs}.}
in Figure \ref{Fig:QuiverdP5}. The low-energy gauge theory phase of this geometry is 5d $SU(2)$ Super Yang-Mills with four fundamental flavours, and the BPS flavour sublattice is
\begin{equation}
\Gamma_f\simeq Q\left(D_5^{(1)} \right),
\end{equation}
whose realization we recall in Appendix \ref{app:dP5}.
\begin{figure}[h]
\begin{center}
\includegraphics[width=.35\textwidth]{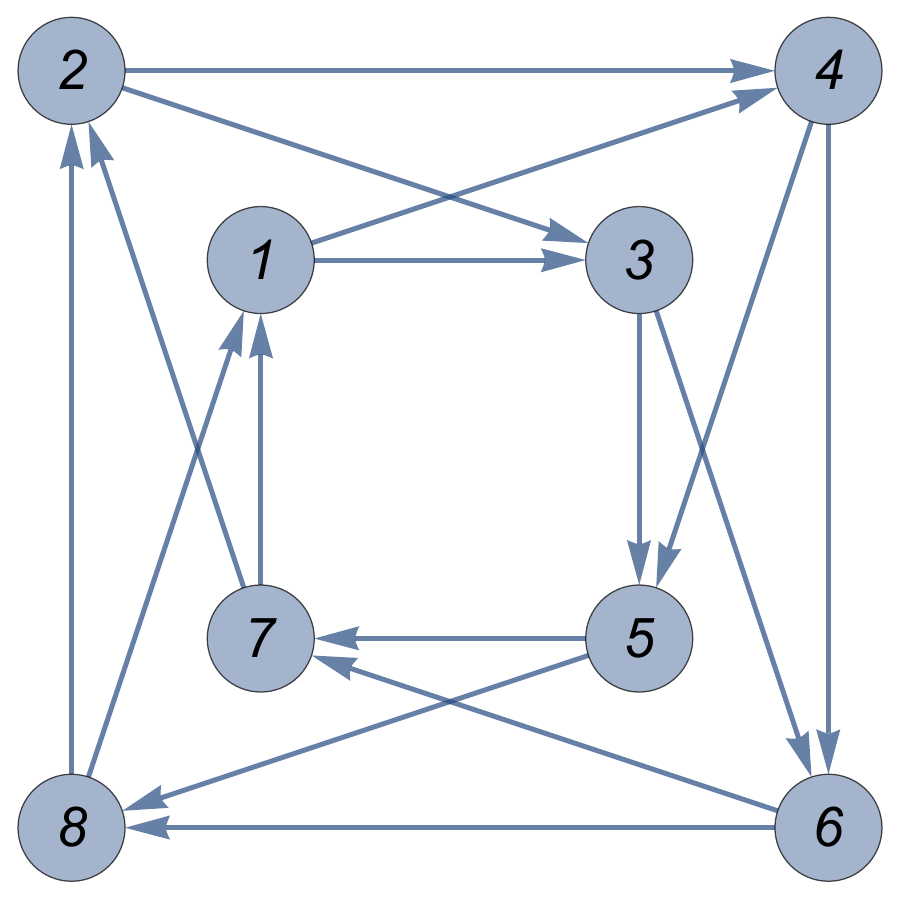}
\end{center}
\caption{BPS quiver for $\text{dP}_5$}
\label{Fig:QuiverdP5}
\end{figure}
Consider the stability condition
\begin{align}
\CC_1^{(alg)}(dP_5) :\qquad Z_1=Z_5=Z+\Lambda_1, \qquad Z_2=Z_6=Z+\Lambda_2, \nonumber\\
 Z_3=Z_7=\bar{Z}+\Lambda_3, \qquad Z_4=Z_8=\bar{Z}+\Lambda_4, \label{eq:CC1algdP5}
\end{align}
with 
\begin{equation}
\Lambda_i\in\mathbb{R},\quad \Lambda_1<\Lambda_2,\,\Lambda_4<\Lambda_3, \qquad 4\Re Z+\Lambda_1+\Lambda_2+\Lambda_3+\Lambda_4=\frac{\pi}{R}, \quad \Im Z>0
\end{equation}
and $\Lambda_i$ sufficiently small so that no wall-crossing happens from the case $\Lambda_1=\Lambda_2=\Lambda_3=\Lambda_4=0$, which we will refer to as the fine-tuned stratum $\CC_1^{(0)}$. In \cite{DelMonte2021} the BPS spectrum for this chamber was shown to be of the form\footnote{There is also another collimation chamber $\CC_2^{(alg)}$, obtained from \eqref{eq:CC1algdP5} through permutation $\pi=(1,3,5,7)(2,4,6,8)$ of the nodes. It corresponds to the algebraic solution of the second collimation chamber $\CC_2$ considered in \cite{DelMonte2021}. For concreteness we will consider only the chamber $\CC_1$, but everything holds for the other chamber as well, after permutation of the labels in the formulas.}:
\begin{equation}\label{eq:SpectrumdP5}
\begin{array}{|c|c|}
	\hline
	\gamma & \Omega(\gamma;y) \\
	\hline\hline
	\gamma_r + k v_1 & 1\\
	-\gamma_r + (k+1) v_1 & 1\\
	\gamma_s + k v_2 & 1\\
	-\gamma_s + (k+1) v_2 & 1\\
	\hline
	\gamma_a+\gamma_b + k \gamma_{D0} & 1\\
	-\gamma_a-\gamma_b + (k+1) \gamma_{D0} & 1\\
	v_1 + k \gamma_{D0} & y+y^{-1}\\
	-v_1 + (k+1) \gamma_{D0} & y+y^{-1}\\
	(k+1)\gamma_{D0} & y^3 + 6y+y^{-1}\\
	\hline
\end{array}
\end{equation}
where $k\geq  0$ and 
\be\label{eq:dP5-charges-domains}
\begin{split}
	& r\in \{1,2,3,4\},\quad s\in \{5,6,7,8\}  \\
	& (a,b)\in \{(1,3), (1,4), (2,3), (2,4), (5,7), (5,8), (6,7), (6,8)\} ,
	\\
	& v_1=\gamma_1+\gamma_2+\gamma_3+\gamma_4,\,\qquad v_2=\gamma_5+\gamma_6+\gamma_7+\gamma_8.
\end{split}
\ee
\begin{figure}
\begin{center}
\begin{tikzpicture}
\draw[fill=SecondBlue,thick] (-1,0) circle (.3);
\draw[fill=SecondBlue,thick] (1,0) circle (.3);
\draw[fill=SecondBlue,thick] (-3,0) circle (.3);
\draw[fill=SecondBlue,thick] (3,0) circle (.3);
\draw[fill=SecondBlue,thick] (-1,2) circle (.3);
\draw[fill=SecondBlue,thick] (1,2) circle (.3);
\draw[-,thick] (-0.7,0) to (0.7,0);
\draw[-,thick] (-2.7,0) to (-1.3,0);
\draw[-,thick] (2.7,0) to (1.3,0);
\draw[-,thick] (-1,0.3) to (-1,1.7);
\draw[-,thick] (1,0.3) to (1,1.7);
\draw[<->,thick] (-1.4,1.8) to (-2.8,0.4);
\draw[<->,thick] (1.4,1.8) to (2.8,0.4);
\node at (-3,0) {$\alpha_1$};
\node at (-1,0) {$\alpha_2$};
\node at (1,0) {$\alpha_3$};
\node at (3,0) {$\alpha_4$};
\node at (-1,2) {$\alpha_0$};
\node at (1,2) {$\alpha_5$};
\node at (-1.8,0.8) {$\pi^2$};
\node at (1.8,0.8) {$\pi^2$};
\end{tikzpicture}
\end{center}
\caption{$D_5^{(1)}$ Dynkin diagram and involution}
\label{Fig:D5Dynkin}
\end{figure}
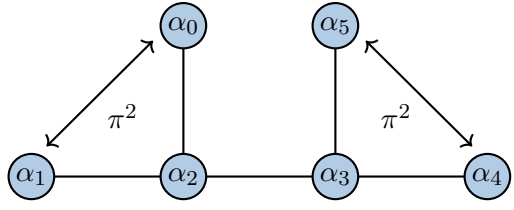
The BPS spectrum \eqref{eq:SpectrumdP5} is invariant under the involution $\pi^2=(1,5)(2,6)(3,7)(4,8) $, which can be identified with the Dynkin diagram automorphism in Figure \ref{Fig:D5Dynkin}. Finally, comparing with the basis of $\Gamma_f$ consisting of $D_5^{(1)}$ simple roots from equation \eqref{eq:D51roots}, we can see that
\begin{align}
\gamma_1+\pi^2(\gamma_1)=\gamma_5+\pi^2(\gamma_5)=\alpha_2, && \gamma_2+\pi^2(\gamma_2)=\gamma_6+\pi^2(\gamma_6)=\alpha_1+\alpha_1+\alpha_2, \nonumber\\
\gamma_3+\pi^2(\gamma_3)=\gamma_7+\pi^2(\gamma_7)=\alpha_3, && \gamma_4+\pi^2(\gamma_4)=\gamma_8+\pi^2(\gamma_8)= \alpha_3+\alpha_4+\alpha_5.
\end{align}
Thus, the assumptions of theorem \ref{thm:ExSo} are verified, and
\begin{align}
\CX_{\gamma_1}=\CX_{\gamma_5}=e^{\frac{Z+\Lambda_1}{\epsilon}}, && \CX_{\gamma_2}=\CX_{\gamma_6}=e^{\frac{Z+\Lambda_2}{\epsilon}}, \nonumber\\
 \CX_{\gamma_3}=\CX_{\gamma_7}=e^{\frac{\bar{Z}+\Lambda_3}{\epsilon}}, && \CX_{\gamma_4}=\CX_{\gamma_8}=e^{\frac{\bar{Z}+\Lambda_4}{\epsilon}}\label{eq:dP5ExSol}
\end{align}
is an exact solution of the TBA equations, as can also be checked by explicit computation.
We will now introduce the cluster Integrable System associated to this geometry, and use it to write down an infinite class of exact solutions starting from \eqref{eq:dP5ExSol}.
\subsubsection*{The q-Painlev\'e VI Cluster Integrable System}
The Cluster Integrable System for this case is the q-Painlev\'e VI equation (symmetry type $D_5^{(1)}$ in Sakai's classification). The multiplicative root variables (Casimirs) are
\begin{align}
	a_0=\CX_2\CX_1^{-1}, && a_1=\CX_6\CX_5^{-1} && a_2=\CX_1\CX_5, \nonumber\\
	a_3=\CX_3\CX_7, && a_4=\CX_4\CX_3^{-1}, && a_5=\CX_8\CX_7^{-1} \label{eq:qPVICas},
\end{align}
with
\begin{equation}
q=\prod_{i=1}^8\CX_i=a_0a_1a_2^2a_3^2a_4a_5.
\end{equation}
The q-Painlev\'e dynamics is described by two log-canonical variables
\begin{align}\label{eq:FGqPVI}
	F=\left(\frac{\CX_7\CX_8}{\CX_3\CX_4} \right)^{\frac{1}{4}}=\left(\frac{a_5}{a_3^2a_4}\right)^{\frac{1}{4}}\CX_7, && G=\left(\frac{\CX_5\CX_6}{\CX_1\CX_2} \right)^{\frac{1}{4}}=\left(\frac{a_1}{a_0a_2^2}\right)^{\frac{1}{4}}\CX_5, && \{F,G\}=1,
\end{align}
evolving through the discrete time evolution
\begin{equation}\label{eq:T1qP6}
\begin{split}
	T_1 & =s_3s_4s_5s_3s_2s_1s_0s_2(2,6)(1,5)(4,8)(7,3) \\
	& = (1,2)(3,4)(5,6)(7,8)\mu_4\mu_8\mu_3\mu_7\mu_2\mu_6\mu_1\mu_5 ,
\end{split}
\end{equation}
which is an affine translation on the root lattice $Q(D_5^{(1)})$. Here $s_i$ denotes the simple reflection along the root $\alpha_i$, and is given in Appendix \ref{App:WeylGroups}. It acts on the (multiplicative) affine roots as
\begin{align}
	T_1(\vec{\alpha})=\left(\alpha_0,\alpha_1,\alpha_2+\delta,\alpha_3-\delta,\alpha_4,\alpha_5 \right), && T_1(\vec{a})=\left(a_0,a_1,qa_2,q^{-1}a_3^{-1},a_4,a_5 \right),
\end{align}
on the BPS charges as
\begin{equation}\label{eq:T1gammadP5}
T_1^n(\vec{\gamma})=\left( \begin{array}{c}
\gamma_1+n\,v_1 \\
\gamma_2+n\,v_1 \\
\gamma_3-n\,v_1 \\
\gamma_4-n\,v_1 \\
\gamma_5+n\,v_2 \\
\gamma_6+n\,v_2 \\
\gamma_7-n\,v_2 \\
\gamma_8-n\,v_2
\end{array} \right),
\end{equation}
and on the X-cluster variables as 
\begin{align}
\begin{cases}
	\overline{\CX_3}=\frac{\left(1+\CX_1^{-1} \right)\left(1+\CX_2^{-1} \right)}{\CX_4\left(1+\CX_5 \right)\left(1+\CX_6 \right)}, \\ \\
	\overline{\CX_4}=\frac{\left(1+\CX_1^{-1} \right)\left(1+\CX_2^{-1} \right)}{\CX_3\left(1+\CX_5 \right)\left(1+\CX_6 \right)}, \\ \\
	\overline{\CX_7}=\frac{\left(1+\CX_5^{-1} \right)\left(1+\CX_6^{-1} \right)}{\CX_8\left(1+\CX_1 \right)\left(1+\CX_2 \right)}, \\ \\
	\overline{\CX_8}=\frac{\left(1+\CX_5^{-1} \right)\left(1+\CX_6^{-1} \right)}{\CX_7\left(1+\CX_1 \right)\left(1+\CX_2 \right)},
\end{cases} &&
\begin{cases}
	\underline{\CX_1}=\frac{\left(1+\CX_7^{-1} \right)\left(1+\CX_8^{-1} \right)}{\CX_2\left(1+\CX_3 \right)\left(1+\CX_4 \right)}, \\ \\
	\underline{\CX_2}=\frac{\left(1+\CX_7^{-1} \right)\left(1+\CX_8^{-1} \right)}{\CX_1\left(1+\CX_3\right)\left(1+\CX_4 \right)}, \\ \\
	\underline{\CX_5}=\frac{\left(1+\CX_3^{-1} \right)\left(1+\CX_4^{-1} \right)}{\CX_6\left(1+\CX_7\right)\left(1+\CX_8 \right)}, \\ \\
	\underline{\CX_6}=\frac{\left(1+\CX_3^{-1} \right)\left(1+\CX_4^{-1} \right)}{\CX_5\left(1+\CX_7 \right)\left(1+\CX_8 \right)},
\end{cases}
\end{align}
where an overline represents the action of $T_1$, while the underline represents the action of $T_1^{-1}$. One can explicitly verify \cite{Bershtein2017} that $F,G$ defined by \eqref{eq:FGqPVI} satisfy the sixth q-Painlev\'e equation (as written in \cite{tsuda2014uc} after the replacement $b_i\mapsto b_i^{-\frac{1}{4}}$)
\begin{align}\label{eq:qPVIeq}
	\overline{F}\,F=\frac{1}{b_5b_6}\frac{(1+b_6G)(1+b_5G)}{(1+b_7G)(1+b_8G)}, &&  \underline{G}\,G=\frac{1}{b_1b_2}\frac{(1+b_1F)(1+b_2F)}{(1+b_3F)(1+b_4F)},
\end{align}
where
\begin{align}
	b_1^4=\frac{a_4}{a_3^2a_5}, && b_2^4=\frac{1}{a_3^2a_4^3a_5}, && b_3^4=\frac{a_3^2a_4}{a_5}, && b_4=a_3^2a_4a_5^3, \\
	b_5^4=\frac{a_0a_2^2}{a_1}, && b_6=a_0a_1^3a_2^2, && b_7^4=\frac{a_0}{a_1a_2^2}, && b_8^4=\frac{1}{a_0^3a_1a_2^2}.
\end{align}
Note that the q-Painlev\'e VI time evolution generates all the states in the BPS spectrum \eqref{eq:SpectrumdP5} whose central charges lie outside the real axis, since it produces the tilting of the positive half-plane associated to the stability condition \eqref{eq:CC1algdP5} (in fact, also to the more general stability condition \eqref{eq:CC1defdP5} that we will introduce below).

\subsubsection*{Algebraic and rational solutions of q-Painlev\'e VI and exact solutions to the TBA}
Algebraic solutions are typically characterized by their invariance under certain B\"acklund transformations known as foldings, implementing appropriate Dynkin diagram automorphisms on the cluster variables (see e.g. \cite{kajiwara2017geometric,Bershtein:2021gkq} and references therein), and they provide the integrable system counterpart of the exact solutions to the TBAs of Theorem \ref{thm:ExSo}.
In particular, if we impose invariance under the involution $\pi^2$, 
we obtain the following algebraic solution of q-Painlev\'e VI
\begin{align}\label{eq:dP5algsol}
	\vec{\CX}^{(dP_5)}_{alg}\left(a_0,a_2,a_3,a_4\right)=\left(\begin{array}{c}
		a_2^{1/2} \\
		a_0a_2^{1/2} \\
		a_3^{1/2} \\
		a_4a_3^{1/2} \\
		a_2^{1/2} \\
		a_0a_2^{1/2} \\
		a_3^{1/2} \\
		a_4a_3^{1/2}	
	\end{array} \right), && q=(a_0a_2a_3a_4)^2,
\end{align}
which coincides with \eqref{eq:dP5ExSol} after appropriate identification between the Casimirs $a_0,a_2,a_3,a_4,q$ and the stability data $Z,\Lambda_1,\Lambda_2,\Lambda_3,\Lambda_4$:
\begin{align}\label{eq:dP5algsol}
  \vec{\CX}^{(dP_5)}_{alg}\left(Z,\Lambda_i\right)=\left(\begin{array}{c}
    e^{\frac{1}{\epsilon}(Z+\Lambda_1)} \\
    e^{\frac{1}{\epsilon}(Z+\Lambda_2)}  \\
    e^{\frac{1}{\epsilon}(Z+\Lambda_3)}  \\
    e^{\frac{1}{\epsilon}(Z+\Lambda_4)}  \\
    e^{\frac{1}{\epsilon}(Z+\Lambda_1)}  \\
    e^{\frac{1}{\epsilon}(Z+\Lambda_2)}  \\
    e^{\frac{1}{\epsilon}(Z+\Lambda_3)}  \\
    e^{\frac{1}{\epsilon}(Z+\Lambda_4)}  
  \end{array} \right), && 4\Re Z+\sum \Lambda_i=\frac{\pi}{R}.
\end{align}
From the algebraic solution \eqref{eq:dP5algsol} it is possible to construct an infinite number of rational solutions by applying B\"acklund transformations, and in particular other affine translations realized as elements of the cluster modular group\footnote{See section 3.2 of \cite{kajiwara2017geometric} for an algorithmic procedure to construct any affine translation as a sequence of the simple reflections \eqref{eq:dP5refl}.}.
The fact that these are all solutions to the q-Painlev\'e equation \eqref{eq:qPVIeq} follows from the commutativity of affine translations: for any translation $T$, one has
\begin{equation}
 T_1T\vec{\CX}_{alg}^{(dP_5)}(a_0,\dots,a_5)=TT_1\vec{\CX}_{alg}^{(dP_5)}(a_0,\dots,a_5)=T\vec{\CX}_{alg}^{(dP_5)}(a_0,a_1,qa_2,q^{-1}a_3,a_4,a_5).
\end{equation}
From the point of view of the BPS Riemann-Hilbert problem, this gives closed-form exact solutions for an infinite family of stability conditions related to each other by affine translations, in terms of the stability data of the original stability condition $\CC_1^{(alg)}$ in \eqref{eq:CC1algdP5}.

\textbf{Example: }
Consider the affine translation
\begin{align}
T_2:=\pi^2s_0s_2s_3s_5s_4s_3s_2s_0, && T_2(\vec{\alpha})=\left(\alpha_0+\delta,\alpha_1-\delta,\alpha_2,\alpha_3,\alpha_4,\alpha_5 \right),
\end{align}
or in terms of multiplicative root variables
\begin{equation}\label{eq:T2Cas}
T_2(\vec{a})=(qa_0,q^{-1}a_1,a_2,a_3,a_4,a_5).
\end{equation}
If we apply this to the algebraic solution \eqref{eq:dP5algsol}, we obtain the following rational solution of q-Painlev\'e VI
{\small\begin{equation}
\begin{split}
\CX_{5,rat}^{(dP_5)}:=T_2(\CX_{5,alg}^{(dP_5)})=\frac{a_2^3 a_3^2 \left(a_2 \left(\left(a_2+1\right) a_3+1\right) a_0^2+1\right) a_4^2 \left(a_2 \left(\left(a_2+1\right) a_3 a_4^2+1\right) a_0^2+1\right)}{\left(a_0^2 a_3 a_2^2+\left(a_3+1\right) a_2+1\right) \left(a_0^2 a_2^2 a_3 a_4^2+a_2 \left(a_3 a_4^2+1\right)+1\right)},
\end{split}
\end{equation}
\begin{equation}
\begin{split}
&\CX_{5,rat}^{(dP_5)}:=T_2(\CX_{5,alg}^{(dP_5)})\\
&=\frac{a_3 \left(a_2 \left(a_3 \left(a_2 a_3^2 \left(\left(a_2 a_0^2+1\right) a_3 a_2+a_2+1\right) a_4^4+\left(a_2+1\right) \left(a_3+1\right) \left(a_2 a_3+1\right) a_4^2+a_2+1\right)+1\right) a_0^2+1\right)}{a_2 \left(a_3 \left(a_0^2 a_2 a_3^2 \left(a_2 \left(\left(a_2+1\right) a_3+1\right) a_0^2+1\right) a_4^4+\left(a_2 a_0^2+1\right) \left(a_3+1\right) \left(a_2 a_3 a_0^2+1\right) a_4^2+a_0^2 a_2+1\right)+1\right)+1},
\end{split}
\end{equation}}
the other cluster variables being obtained by a combination of \eqref{eq:T2Cas} and \eqref{eq:qPVICas}. This is an exact solution for the BPS Riemann-Hilbert problem corresponding to the transformed stability condition $\tilde{Z}_i:=T_2(Z_i)$, with
\begin{align}
T_2(\CC_1^{(alg)}): && T_2(Z_1)=Z_1-2\Re Z_1-\Lambda_3-\Lambda_4,\quad T_2(Z_2)=Z_2+2\Re Z_2+\Lambda_3+\Lambda_4, \nonumber \\
&& T_2(Z_3)=Z_3,\quad T_2(Z_4)=Z_4, \quad T_2(Z_5)=Z_1+2\Re Z_1+\Lambda_3+\Lambda_4, \nonumber\\
&& T_2(Z_6)=Z_6-2\Re Z_6-\Lambda_4-\Lambda_4, \quad T_2(Z_7)=Z_7,\quad T_2(Z_8)=Z_8,
\end{align}
with $Z_i$ given by \eqref{eq:CC1algdP5}. This stability condition does not satisfy the assumptions of Theorem \ref{thm:ExSo}, and so its solution receives all-order corrections in $\epsilon$. Nontheless, an exact solution can be given starting from the algebraic one.
\begin{remark}
	The rational solutions have a nice characterization from a resurgent standpoint: recalling that $a_i=\exp(Z_{\alpha_i}/\epsilon)$, they are trans-series with no perturbative corrections around all the infinite saddles. The algebraic solution is the one consisting of only one trans-monomial. To the author's knowledge, from the point of view of the BPS spectral problem the existence of such solutions has never been pointed out before.
\end{remark}

\subsubsection*{q-Painlev\'e VI equation and TBA equation}
Up to now we used the fact that the algebraic solution satisfies the assumptions of Theorem \ref{thm:ExSo} to construct relatively simple solutions (algebraic and rational in the multiplicative root variables) of the BPS Riemann-Hilbert Problem. Before moving forward, let us stress that it is possible show that
any solution to \eqref{eq:TBA-conformal} corresponding to the spectrum \eqref{eq:SpectrumdP5} must also solve the q-Painlev\'e equation \eqref{eq:qPVIeq}. This is analogous to the perhaps more familiar statement that solutions of the TBA equations for finite chambers can be reformulated as a Y-system, such as those appearing in the study of 2d integrable QFTs \cite{Zamolodchikov:1991et,Ravanini:1992fi,Kuniba:1993cn,Cecotti:2014zga}.

We will only sketch the derivation, as the argument is simply an adaptation of the one used in \cite{DelMonte2022a} for local $\mathbb{P}^1\times\mathbb{P}^1$: consider the following deformation of the family of stability conditions $\CC_1^{(alg)}$\footnote{This is only one neighborhood of $\CC_1^{(alg)}$ that is particularly convenient for the present discussion, but more general ones can be considered within the same chamber.}:
\begin{align}
\CC_1^{(def)}(dP_5) : \qquad Z_1=Z+\Lambda_1, \quad Z_2=Z+\Lambda_2, \quad Z_3=\bar{Z}+\Lambda_3, \quad Z_4=\bar{Z}+\Lambda_4, \nonumber\\
 Z_5=Z+\Lambda_1+\Lambda_5, \quad Z_6=Z+\Lambda_2+\Lambda_6, \quad Z_7=\bar{Z}+\Lambda_3+\Lambda_7, \quad Z_8=\bar{Z}+\Lambda_4+\Lambda_8,  \label{eq:CC1defdP5}
\end{align}
with
\begin{gather}
\Lambda_i\in\mathbb{R},\qquad\max(\Lambda_1,\,\Lambda_1+\Lambda_5)<\min(\Lambda_2,\,\Lambda_2+\Lambda_6),\quad\max(\Lambda_4,\Lambda_4+\Lambda_8)<\min(\Lambda_3,\Lambda_3+\Lambda_7), \nonumber \\{}
4\Re Z+\Lambda_1+\Lambda_2+\Lambda_3+\Lambda_4=\frac{\pi}{R},\qquad \Lambda_5+\Lambda_6+\Lambda_7+\Lambda_8=0,\qquad \Im Z>0,
\end{gather}
and perform the following homotopy in the space of stability conditions:
\begin{align}
Z_{\gamma_i}(s)=(1-s)Z_{\gamma_i}+s\widetilde{Z}_{\gamma_i}.
\end{align}
Here
\begin{align}
\widetilde{Z}_{\gamma_{1,5}}=Z_{\gamma_{1,5}}+Z_{v_1}, \quad \widetilde{Z}_{\gamma_{2,6}}=Z_{\gamma_{2,6}}+Z_{v_2}, \\
\widetilde{Z}_{\gamma_{3,7}}=Z_{\gamma_{3,7}}-Z_{v_2}, \quad \widetilde{Z}_{\gamma_{4,8}}=Z_{\gamma_{4,8}}-Z_{v_1},
\end{align}
so that we are shifting the towers of central charges with positive imaginary part to the right, and the towers with negative imaginary part to the left. Along this path in the moduli space of stability conditions, the integration contours of the TBA rotate, and the phase of $\epsilon$ is crossed first by the following paths, in sequence: $\ell_{\gamma_1}$, $\ell_{\gamma_5}$, $\ell_{\gamma_2} $, $\ell_{\gamma_6}$, $\ell_{\gamma_5+\gamma_6+\gamma_8}$,  $\ell_{\gamma_1+\gamma_2+\gamma_4}$, $\ell_{\gamma_1+\gamma_2+\gamma_3}$, $\ell_{\gamma_5+\gamma_6+\gamma_7}$. This has the effect of the following sequence of mutations on the solution to the TBA equations:
\begin{equation}
\textbf{m}_1=\mu_4\mu_8\mu_3\mu_7\mu_2\mu_6\mu_1\mu_5,
\end{equation}
which is the same (up to a permutation that only exchanges labels of mutually local central charges and does not affect the result) as the sequence of mutations entering in the affine translation \eqref{eq:T1qP6} defining the q-Painlev\'e VI time evolution. Since the q-Painlev\'e VI equation follows only from the mutation rules, it follows that solutions of the GMN TBA equation \eqref{eq:TBA-conformal} in the collimation chamber $\CC_1^{(def)}$ are also solutions of the q-Painlev\'e VI equation \eqref{eq:qPVIeq}. The explicit expression for the general solution in terms of dual Nekrasov partition functions was obtained in \cite{Jimbo:2017}, so that one has an explicit formula for the solution of the TBA equation \eqref{eq:TBA-conformal} in terms of Nekrasov functions. We will not go into a detailed description of the general solution, which can be studied applying the methods of \cite{DelMonte2022a} to the solution of \cite{Jimbo:2017}. 

\begin{remark}
	As already noted at the beginning of this section, the q-Painlev\'e equation \eqref{eq:qPVIeq} can be regarded as the analogue of a Y-system for the TBA equations \eqref{eq:TBA-conformal}. Interestingly, the q-Painlev\'e tau functions, which can be written as dual 5d gauge theory partition functions \cite{Bershtein2017,Bonelli2017,Jimbo:2017,Bershtein:2018srt,Bershtein2018,Matsuhira2018}, are related to the $\CX$-cluster variables by
	\begin{equation}
		\CX_i=y_i\prod_{j=1}^{|\CQ|}\tau_j^{B_{ji}},
	\end{equation}
	where $y_i$ are cluster algebra coefficients. This expression seems to point out that the q-Painlev\'e tau functions should be the analogue of Q-functions for the TBA equation \cite{kirillov1990representations,kedem2008q,francesco2010q} .
\end{remark}

\subsection{Local $dP_4$}

Our strategy to study local $dP_4$ will be to decouple appropriate nodes of the BPS quiver \ref{Fig:QuiverdP5} to obtain the desired $dP_4$ quiver (see model 8a in \cite{Hanany:2012hi}, with the same remark about del Pezzos and pseudo del Pezzos as in footnote \ref{fn:PseudodP}). This procedure is called "Higgsing" in the physics literature \cite{Closset:2019juk}, and it was mostly studied in the context of brane tilings/dimer models associated to toric geometries \cite{Franco:2005rj,Davey:2009bp,Hanany:2012hi,Franco:2017jeo}. In our present case, it corresponds to sending to infinity the mass of a hypermultiplet in the low-energy $SU(2)$ $N_f=4$ gauge theory phase of local $dP_4$, degenerating to the $SU(2)$ $N_f=3$ low-energy gauge theory phase of local $dP_4$.

We start from the collimation chamber $\CC_1$ of $dP_5$ and send $Z_5,Z_7\rightarrow\infty$ while keeping $Z_5+Z_7$ finite. This can be realized by sending $\Lambda_5\rightarrow-\infty$, $\Lambda_7\rightarrow+\infty$ in \eqref{eq:CC1defdP5}while keeping $\tilde{\Lambda}_1:=\Lambda_5+\Lambda_7$ finite.
Graphically, the nodes $\gamma_5,\,\gamma_7$ of the quiver merge, and the number of arrows from the $i$-th node of the quiver to the new node $\gamma_5+\gamma_7$ are equal to $B_{i5}+B_{i7}$. 

Before detailing the result of this degeneration procedure, let us remark that in principle when taking this limit an infinite number of wall-crossings happen, possibly producing new wild states. However, as argued in \cite{Closset:2019juk} on the basis of quiver representation theory, in terms of which the limit amounts to imposing that the arrow between the nodes 5 and 7 is an isomorphism, these states will be unstable in the limit, so that they will not contribute to the BPS spectrum. In practice, the argument from \cite{Closset:2019juk} is corroborated by the fact that the BPS spectrum computed by degeneration of \eqref{eq:SpectrumdP5} coincides with the one computed by using the mutation method for the limiting stability condition \eqref{eq:CC1defdP4} below.
After relabeling 
\begin{align}
\gamma_5^{(dP_4)}:=\gamma_6^{(dP_5)}, && \gamma_6^{(dP_4)}:=\gamma_5^{(dP_5)}+\gamma_7^{(dP_5)}, && \gamma_7^{(dP_4)}:=\gamma_8^{(dP_5)},
\end{align}
the quiver takes the form in Figure \ref{Fig:QuiverdP4}, which is a BPS quiver associated to local $dP_4$ and to the q-Painlev\'e equation of symmetry type $A_4^{(1)}$. The stability condition \eqref{eq:CC1defdP5} becomes
\begin{align}
\CC_1^{(def)}(dP_4) : \qquad Z_1=Z+\Lambda_1, \quad Z_2=Z+\Lambda_2, \quad Z_3=\bar{Z}+\Lambda_3, \quad Z_4=\bar{Z}+\Lambda_4, \nonumber\\
Z_5=Z+\Lambda_2+\Lambda_6 , \quad Z_6=2\Re Z+\Lambda_1+\Lambda_3+\tilde{\Lambda}_1, \quad Z_7=\bar{Z}+\Lambda_4+\Lambda_8,  \label{eq:CC1defdP4}
\end{align}
with
\begin{gather}
\Lambda_i\in\mathbb{R},\qquad\Lambda_1<\min(\Lambda_2,\,\Lambda_2+\Lambda_6),\quad\max(\Lambda_4,\Lambda_4+\Lambda_8)<\Lambda_3 \nonumber \\{}
4\Re Z+\Lambda_1+\Lambda_2+\Lambda_3+\Lambda_4=\frac{\pi}{R},\qquad \tilde{\Lambda}_1+\Lambda_6+\Lambda_8=0.
\end{gather}
\begin{figure}[h]
\begin{center}
\includegraphics[width=.35\textwidth]{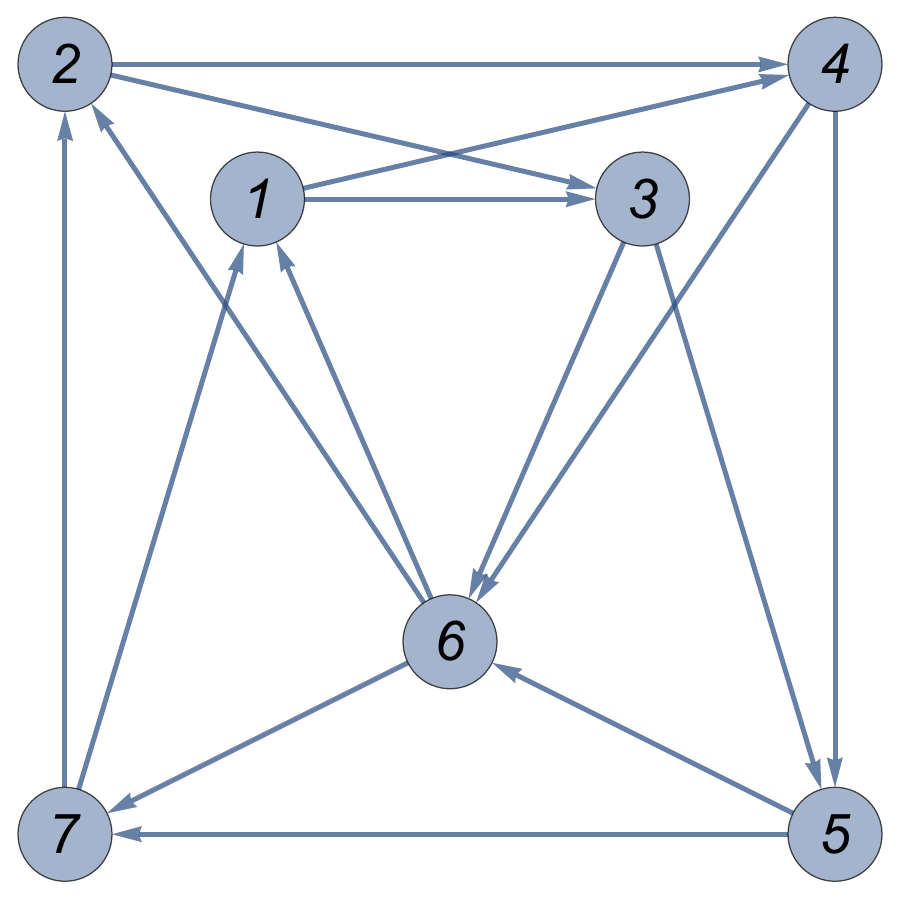}
\end{center}
\caption{BPS quiver for $\text{dP}_4$}
\label{Fig:QuiverdP4}
\end{figure}
Indeed, after the limit, the rank of the flavour sublattice is reduced from 6 to 5, and we can choose the following basis:
\begin{equation}\label{eq:Roots5to4}
\alpha_0^{(dP_4)}:=\alpha_2^{(dP_5)}+\alpha_3^{(dP_5)}, \quad\alpha_1^{(dP_4)}:=\alpha_0^{(dP_5)}, \quad \alpha_2^{(dP_4)}=\alpha_2^{(dP_5)}+\alpha_1^{(dP_5)}, \nonumber
\end{equation}
\begin{equation}
\alpha_3^{(dP_4)}:=\alpha_3^{(dP_5)}+\alpha_5^{(dP_5)}, \qquad \alpha_4^{(dP_4)}=\alpha_4^{(dP_5)},
\end{equation}
which in terms of the BPS charges associated to the quiver $\ref{Fig:QuiverdP4}$ reads (omitting the superscript $dP_4$ when not necessary)
\begin{equation}
\alpha_0=\gamma_1+\gamma_3+\gamma_6, \quad \alpha_1=\gamma_2-\gamma_1, \quad \alpha_2=\gamma_1+\gamma_5, \nonumber
\end{equation}
\begin{equation}
\alpha_3=\gamma_3+\gamma_7, \qquad \alpha_4=\gamma_4-\gamma_3.
\end{equation}
Using the relation \eqref{eq:Roots5to4}, together with the $D_5^{(1)}$ Cartan matrix \eqref{eq:CartanD5}, one finds that the intersection pairing between the basis $\alpha_j$ of $\Gamma_f$ for local $dP_4$ coincides with the Cartan matrix of $A_4^{(1)}$ \eqref{eq:CartanA4}, as expected from the degeneration diagram \ref{Fig:Sakai}. 
The limit can be straightforwardly applied to the BPS spectrum \eqref{eq:SpectrumdP5} as well, by keeping only states from \eqref{eq:SpectrumdP5} whose central charges remain finite: 
\begin{equation}\label{eq:SpectrumdP4}
\begin{array}{|c|c|}
	\hline
	\gamma & \Omega(\gamma;y) \\
	\hline\hline
	\gamma_r + k v_1 & 1\\
	-\gamma_r + (k+1) v_1 & 1\\
	\gamma_s + k v_2 & 1\\
	-\gamma_s + (k+1) v_2 & 1\\
	\hline
	\gamma_a+\gamma_b + k \gamma_{D0} & 1\\
	-\gamma_a-\gamma_b + (k+1) \gamma_{D0} & 1\\
	\gamma_6+k\gamma_{D0} & 1 \\
	-\gamma_6+(k+1)\gamma_{D0} & 1 \\
	v_1 + k \gamma_{D0} & y+y^{-1}\\
	-v_1 + (k+1) \gamma_{D0} & y+y^{-1}\\
	(k+1)\gamma_{D0} & y^3 + 6y+y^{-1}\\
	\hline
\end{array}
\end{equation}
where $k\geq  0$ and 
\be\label{eq:dP5-charges-domains}
\begin{split}
	& r\in \{1,2,3,4\},\quad s\in \{5,7\}  \\
	& (a,b)\in \{(1,3), (1,4), (2,3), (2,4), (5,7)\} ,
	\\
	& v_1=\gamma_1+\gamma_2+\gamma_3+\gamma_4,\,\qquad v_2=\gamma_5+\gamma_6+\gamma_7.
\end{split}
\ee
In fact, one can also find this spectrum by directly by applying the mutation method to the stability condition \eqref{eq:CC1defdP4}. 
On the new basis of BPS charges, the discrete time evolution \eqref{eq:T1gammadP5} degenerates to
\begin{equation}\label{eq:T1gammadP4}
T_1^n(\vec{\gamma})=\left( \begin{array}{c}
\gamma_1+n\,v_1 \\
\gamma_2+n\,v_1 \\
\gamma_3-n\,v_1 \\
\gamma_4-n\,v_1 \\
\gamma_5+n\,v_2 \\
\gamma_6 \\
\gamma_7-n\,v_2 
\end{array} \right),
\end{equation}
which is the following affine translation on the (multiplicative) root variables
\begin{align}
T_1(\vec{\alpha})=\left(\alpha_0,\,\alpha_1,\,\alpha_2+\delta,\,\alpha_3-\delta,\,\alpha_4 \right), && T_1(\vec{a})=\left(a_0,\,a_1,\,q\,a_2,\,q^{-1}a_3,\,a_4 \right),
\end{align}
and can be realized in terms of mutation as
\begin{equation}
T_1=(1,2)(3,4)(5,6,7)\mu_6\mu_3\mu_4\mu_5\mu_2\mu_1.
\end{equation}
This translation also coincides with the time evolution of the corresponding q-Painlev\'e V equation (symmetry type $A_4^{(1)}$) \cite{tsuda2006tau,Bershtein2017}.
The integrable system admits no folding \cite{Bershtein:2021gkq}, and correspondigly we do not have a semiclassically exact solution to the TBAs. Instead, solutions of the TBA equation \eqref{eq:TBA-conformal} in the collimation chamber \eqref{eq:SpectrumdP4} will be described by nontrivial solutions of the corresponding q-Painlev\'e equation, by an analogous argument to the one at the end of the previous section.

\subsection{Local $dP_3$}

To decouple one more flavour, we follow the same procedure as before and take $Z_1,Z_3\rightarrow\infty$ in the stability condition \eqref{eq:CC1defdP4}, while keeping $Z_1+Z_3$ finite, i.e. take $\Lambda_1\rightarrow-\infty$, $\Lambda_3\rightarrow+\infty$ while fixing $\tilde{\Lambda}_2:=\Lambda_1+\Lambda_3 $ in \eqref{eq:CC1defdP4}.
After relabeling
\begin{align}
\gamma_2^{(dP_3)}:=\gamma_1^{(dP_4)}+\gamma_3^{(dP_4)}, && \gamma_1^{(dP_3)}:=\gamma_2^{(dP_4)} , && \gamma_3^{(dP_3)}:=\gamma_4^{(dP_4)} \nonumber\\
\gamma_4^{(dP_3)}:=\gamma_5^{(dP_4)}, && \gamma_5^{(dP_3)}:=\gamma_6^{(dP_4)}, && \gamma_6^{(dP_3)}:=\gamma_7^{(dP_4)},
\end{align}
the resulting quiver for local $dP_3$ is in Figure \ref{Fig:QuiverdP3},
\begin{figure}[b]
\begin{center}
\includegraphics[width=.35\textwidth]{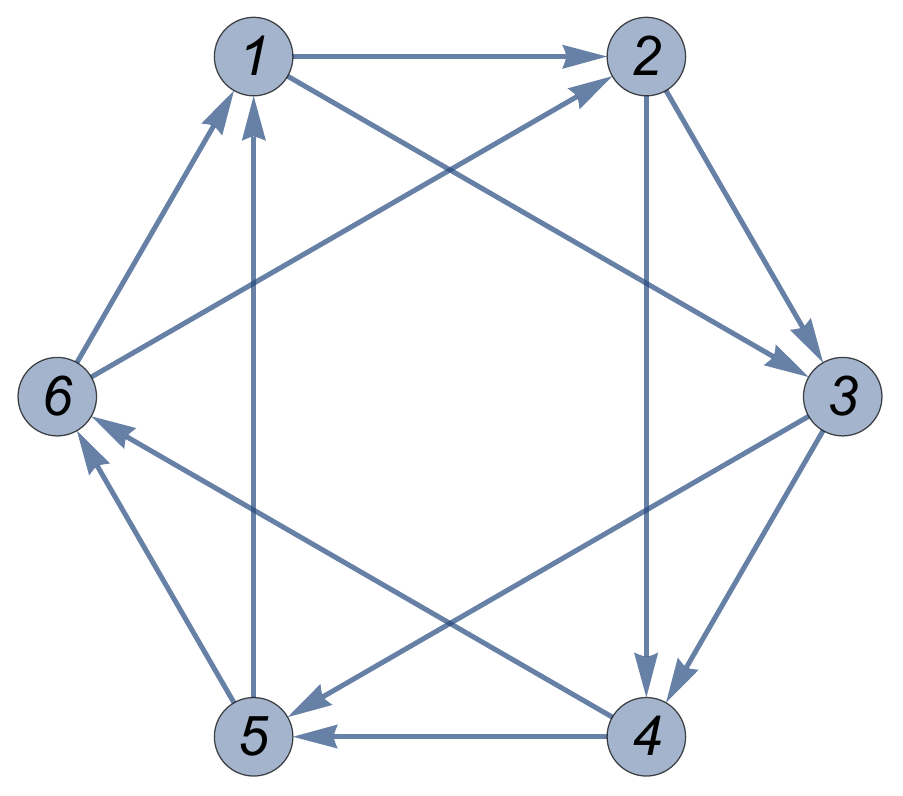}
\end{center}
\caption{BPS quiver for $\text{dP}_3$}
\label{Fig:QuiverdP3}
\end{figure}
and the limiting stability condition 
\begin{align}
\CC_1^{(def)}(dP_3) : && Z_1=Z+\Lambda_2, \qquad Z_2=2\Re Z+\tilde{\Lambda}_2, \qquad Z_3=\bar{Z}+\Lambda_4, \nonumber\\
&&Z_4=Z+\Lambda_2+\Lambda_6 , \quad Z_5=2\Re Z+\tilde{\Lambda}_1+\tilde{\Lambda}_2, \quad Z_6=\bar{Z}+\Lambda_4+\Lambda_8,  \label{eq:CC1defdP3}
\end{align}
with
\begin{gather}
\Lambda_i\in\mathbb{R}, \qquad 4\Re Z+\tilde{\Lambda}_2+\Lambda_2+\Lambda_4=\frac{\pi}{R},\qquad \tilde{\Lambda}_1+\Lambda_6+\Lambda_8=0,
\end{gather}
deforming the fine-tuned stability condition of \cite{DelMonte2021}, which is recovered when $\Lambda_i=\tilde{\Lambda}_i=0$.

We can choose the following basis for the flavour lattice (omitting $dP_3$ labels for convenience):
\begin{align}
\alpha_0=\gamma_3+\gamma_6=\alpha_3^{(dP_4)}+\alpha_4^{(dP_4)} , && \alpha_1=\gamma_1+\gamma_4=\alpha_1^{(dP_4)}+\alpha_2^{(dP_4)}, && \alpha_2=\gamma_2+\gamma_5=\alpha_0^{(dP_4)},\nonumber
\end{align}
\begin{align}\label{eq:dP3-beta-roots}
\beta_0=\gamma_2+\gamma_4+\gamma_6=\alpha_2^{(dP_4)}+\alpha_3^{(dP_4)}, && \beta_1=\gamma_1+\gamma_3+\gamma_5=\alpha_0^{(dP_4)}+\alpha_1^{(dP_4)}+\alpha_4^{(dP_4)}.
\end{align}
As before, we introduce multiplicative root variables by $a_i:=\CX_{\alpha_i}$, $b_i:=\CX_{\beta_i}$. 
The symmetric pairing of the flavour charges \eqref{eq:dP3-beta-roots} coincides with the $(A_2+A_1)^{(1)}$ Cartan matrix \eqref{eq:CartanA2A1} encoded in the Dynkin diagrams in Figure \ref{Fig:A1Dynkin2}.
\begin{figure}
\begin{center}
\begin{subfigure}{.5\textwidth}
\centering
\begin{tikzpicture}

\draw[fill=SecondBlue,thick] (-1,0) circle (.3);
\draw[fill=SecondBlue,thick] (1,0) circle (.3);
\draw[-,thick] (-0.7,0.1) to (0.7,0.1);
\draw[-,thick] (-0.7,-0.1) to (0.7,-0.1);
\draw[->,thick] (-0.8,0.4) to[out=30,in=150] (0.8,0.4);
\draw[<-,thick] (-0.8,-0.4) to[out=330,in=210] (0.8,-0.4);
\node at (-1,0) {$\beta_0$};
\node at (1,0) {$\beta_1$};
\node at (0,1) {$\pi$};
\node at (0,-1) {$\pi$};

\end{tikzpicture}
\end{subfigure}\hfill
\begin{subfigure}{.5\textwidth}
\centering
\begin{tikzpicture}

\draw[fill=SecondBlue,thick] (-1,0) circle (.3);
\draw[fill=SecondBlue,thick] (1,0) circle (.3);
\draw[fill=SecondBlue,thick] (0,1.72) circle (.3);
\draw[-,thick] (-0.7,0) to (0.7,0);
\draw[-,thick] (-1,0.3) to (-0.3,1.72);
\draw[-,thick] (1,0.3) to (0.3,1.72);
\draw[<-,thick] (-1.2,0.4) to[out=30+90,in=180] (-0.4,1.82);
\draw[->,thick] (-0.8,-0.4) to[out=330,in=210] (0.8,-0.4);
\draw[->,thick] (1.2,0.4) to[out=-30+90,in=0] (0.4,1.82);
\node at (-1,0) {$\alpha_2$};
\node at (1,0) {$\alpha_1$};
\node at (0,1.72) {$\alpha_0$};
\node at (-1.5,1.5) {$\pi$};
\node at (1.5,1.5) {$\pi$};
\node at (0,-1) {$\pi$};

\end{tikzpicture}
\end{subfigure}
\end{center}
\caption{$(A_2+A_1)^{(1)}$ Dynkin diagrams and automorphisms. $\pi^2=id$ on the $A_1^{(1)}$ sublattice, while $\pi^3=id$ on the $A_2^{(1)}$ sublattice. The two combine to give the $\mathbb{Z}_6$ permutation symmetry of the quiver.}
\label{Fig:A1Dynkin2}
\end{figure}
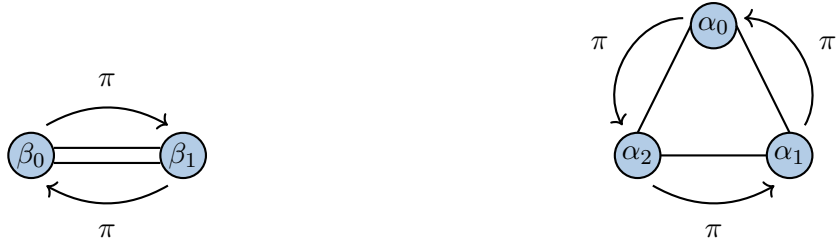
The BPS spectrum obtained from the decoupling procedure coincides with the one computed in \cite{DelMonte2021}, and is given by
\begin{equation}\label{eq:SpectrumdP3}
\begin{array}{|c|c|}
	\hline
	\gamma & \Omega(\gamma;y) \\
	\hline\hline
	\gamma_r + k v_1 & 1\\
	-\gamma_r + (k+1) v_1 & 1\\
	\gamma_s + k v_2 & 1\\
	-\gamma_s + (k+1) v_2 & 1\\
	\hline
	\gamma_t +k\gamma_{D0} & 1 \\
	-\gamma_t +(k+1)\gamma_{D0} & 1 \\
	\gamma_a+\gamma_b +k\gamma_{D0} & 1 \\
	-\gamma_a-\gamma_b +(k+1)\gamma_{D0} & 1 \\
	v_1+k\gamma_{D0} & y+y^{-1} \\
	-v_1+(k+1)\gamma_{D0} & y+y^{-1} \\
	(k+1)\gamma_{D0} & y^3 + 4y+y^{-1}\\
	\hline
\end{array}
\end{equation}
with $k\geq  0$,
\begin{equation}
	r\in\{1,3\}
	\qquad 
	s\in\{4,6\}{}
	\qquad 
	t\in\{2,5\}
	\qquad
	(a,b) \in  \{(1,3), (4,6)\}\,,
\end{equation}
and
\begin{align}
v_1=\gamma_1+\gamma_2+\gamma_3, && v_2=\gamma_4+\gamma_5+\gamma_6.
\end{align}
The corresponding cluster integrable system is the q-Painlev\'e $III_1$ equation (symmetry type $(A_2+A_1)^{(1)}$ in Sakai's classification), with discrete time evolution
\begin{align}\label{eq:T1dP3}
T_1=\pi^2s_2s_1, && T_1(\vec{a},\vec{b})=(q^{-1}a_0,qa_1,a_2,b_0,b_1).
\end{align}
The Dynkin diagram \ref{Fig:A1Dynkin2} now has a $\mathbb{Z}_2$ involution realized on the BPS charges as the permutation
\begin{equation}
\pi^3=(1,4)(2,5)(3,6),
\end{equation}
with $\pi$  given in the cluster realization of the affine Weyl group \eqref{eq:dP3Dih6}.
We are again in a setting where the assumptions of Theorem \ref{thm:ExSo} hold, since when $\Lambda_6=\Lambda_8=0$, $\tilde{\Lambda}_2=\Lambda_1+\tilde{\Lambda}_1$ both the stability condition \eqref{eq:CC1defdP3} and the spectrum \eqref{eq:SpectrumdP3} are invariant under $\pi^3$. The exact solution of the TBA equation is
\begin{align}\label{eq:exsolP3}
\CX_1=\CX_4=e^{\frac{Z+\Lambda_2}{\epsilon}}, && \CX_2=\CX_5=e^{\frac{2\Re Z+\tilde{\Lambda}_2}{\epsilon}}, && \CX_3=\CX_6=e^{\frac{\bar{Z}+\Lambda_4}{\epsilon}}, 
\end{align}
coinciding with the q-Painlev\'e $III_3$ algebraic solution 
\begin{align}\label{eq:algdP3}
\CX_1=\CX_4=a_1^{\frac{1}{2}}, && \CX_2=\CX_5=a_2^{\frac{1}{2}}, && \CX_3=\CX_6=a_0^{\frac{1}{2}}, && a_0a_1a_2=q.
\end{align}
\begin{remark}
Note that the algebraic solution of $dP_3$ can be obtained by direct decoupling of the algebraic solution of $dP_5$. This is because if one tries to take the limit $Z_5,Z_7\rightarrow\infty$, $Z_5+Z_7$ finite that takes $dP_5$ to $dP_4$, the $\mathbb{Z}_2$ symmetry of the stability condition \eqref{eq:CC1algdP5} corresponding to the algebraic solution \eqref{eq:dP5algsol} requires also to simultaneously take the limit $Z_1,Z_3\rightarrow\infty$ with $Z_1+Z_3$ finite, that brings $dP_4$ to $dP_3$. In terms of multiplicative roots/K\"ahler parameters, the limit reads 
\begin{align}
a_2^{(dP_5)},a_4^{(dP_5)}\rightarrow\infty,&&a_0^{(dP_5)},a_3^{(dP_5)}\rightarrow0, 
\end{align}
while keeping finite
\begin{align}
a_1^{(dP_3)}:=a_0^{(dP_5)}(a_2^{(dP_5)})^2, && a_2^{(dP_3)}:=a_2^{(dP_5)}a_3^{(dP_5)}, && a_3^{(dP_3)}:=a_4^{(dP_5)}(a_3^{(dP_5)})^2.
\end{align}
\end{remark}
As it happened for the case of local $dP_5$, there are nontrivial B\"acklund transformation, consisting of flows $T_2,T_3$ on the $A_2^{(1)}$ sublattice (equation \eqref{eq:T1T2T3}), and of the flow $T_4$ on the $A_1^{(1)}$ sublattice (equation \eqref{eq:T4}). Rational solutions are constructed from the elementary one by action of B\"acklund transformations. The flows $T_2$, $T_3$ act rather simply:
\begin{align}
	T_2(\CX_i(a_0,a_1,a_2))=\CX_i(a_0,q^{-1}a_2,qa_3), && T_3(\CX_i(a_0,a_1,a_2))=\CX_i(qa_0,a_1,q^{-1}a_2).
\end{align}
On the other hand, the action of $T_4=(4,6)\mu_2\mu_4\mu_6\mu_2(4,5,1,2,3)$ gives quite nontrivial rational solutions. For example,
\begin{align}
	T_4(\CX_1)=a_1a_3^{\frac{1}{2}}\left(\frac{ 1+a_2^{\frac{1}{2}}\left(1+a_1^{\frac{1}{2}}\right) }{1+a_1^{\frac{1}{2}} \left(1+a_3^{\frac{1}{2}}\right)}\right), 
	&& 
	T_4(\CX_2)=a_1^{-\frac{1}{2}}\left(\frac{1+a_2^{\frac{1}{2}}\left(1+a_1^{\frac{1}{2}}\right) }{1+ a_3^{\frac{1}{2}}\left(1+a_2^{\frac{1}{2}}\right) }\right)
\end{align}
\begin{align}
	T_4(\CX_3)=a_2^{\frac{1}{2}} a_3\left(\frac{ 1+a_1^{\frac{1}{2}} \left(1+a_3^{\frac{1}{2}}\right)}{1+a_3^{\frac{1}{2}}\left(1+a_2^{\frac{1}{2}}\right) }\right), 
	&& 
	T_4(\CX_4)=a_3^{-\frac{1}{2}}\left(\frac{1+a_1 ^{\frac{1}{2}}\left(1+a_3^{\frac{1}{2}}\right)}{1+a_2^{\frac{1}{2}}\left(1+a_1^{\frac{1}{2}}\right)}\right)
\end{align}
\begin{align}
	T_4(\CX_5)=a_1^{\frac{1}{2}} a_2\left(\frac{1+ a_3^{\frac{1}{2}}\left(1+a_2^{\frac{1}{2}}\right) }{1+a_2^{\frac{1}{2}}\left(1+a_1^{\frac{1}{2}}\right) }\right), 
	&& 
	T_4(\CX_6)=a_2^{-\frac{1}{2}}\left(\frac{1+a_3^{\frac{1}{2}}\left(1+a_2^{\frac{1}{2}}\right) }{1+a_1^{\frac{1}{2}} \left(1+a_3^{\frac{1}{2}}\right)}\right)
\end{align}
is another solution. An infinite number of rational solutions for the q-Painlev\'e $III_3$  flow generated by $T_1$ can be generated by successive application of $T_2,T_3,T_4$ to the "seed" solution \eqref{eq:algdP3}, yielding an infinite number of exact solutions to the BPS RHP in the respective transformed chambers.

\subsection{Local $dP_2$}

To degenerate further to local $dP_2$, we send $Z_4,\,Z_5\rightarrow\infty$ 
While keeping $Z_4+Z_5$ finite, which can be achieved by sending $\Lambda_6\rightarrow+\infty$, $\tilde{\Lambda}_1\rightarrow-\infty$ while keeping finite $\Lambda_6+\tilde{\Lambda}_1=-\Lambda_8$. After relabeling
\begin{align}
\gamma_1^{(dP_2)}:=\gamma_3^{(dP_3)}, && \gamma_2^{(dP_2)}:=\gamma_4^{(dP_3)}+\gamma_5^{(dP_3)}, && \gamma_3^{(dP_2)}:=\gamma_6^{(dP_3)}, \nonumber
\end{align}
\begin{align}
\gamma_4^{(dP_2)}:=\gamma_1^{(dP_3)}, && \gamma_5^{(dP_2)}:=\gamma_2^{(dP_3)}
\end{align}
one obtains the local $dP_2$ quiver in Figure \ref{Fig:QuiverdP2}. The stability condition becomes
\begin{align}
\CC_1^{(def)}(dP_2) : && Z_1=\bar{Z}+\Lambda_4, \qquad Z_2=Z+2\Re Z+\tilde{\Lambda}_2+\Lambda_2-\Lambda_8,\nonumber\\
&& \qquad Z_3=\bar{Z}+\Lambda_4+\Lambda_8, \quad Z_4=Z+\Lambda_2 , \quad Z_5=2\Re Z+\tilde{\Lambda}_2,  \label{eq:CC1defdP2}
\end{align}
with
\begin{gather}
\Lambda_i\in\mathbb{R}, \qquad 4\Re Z+\tilde{\Lambda}_2+\Lambda_2+\Lambda_4=\frac{\pi}{R}.
\end{gather}
\begin{figure}[h]
\begin{center}
\includegraphics[width=.35\textwidth]{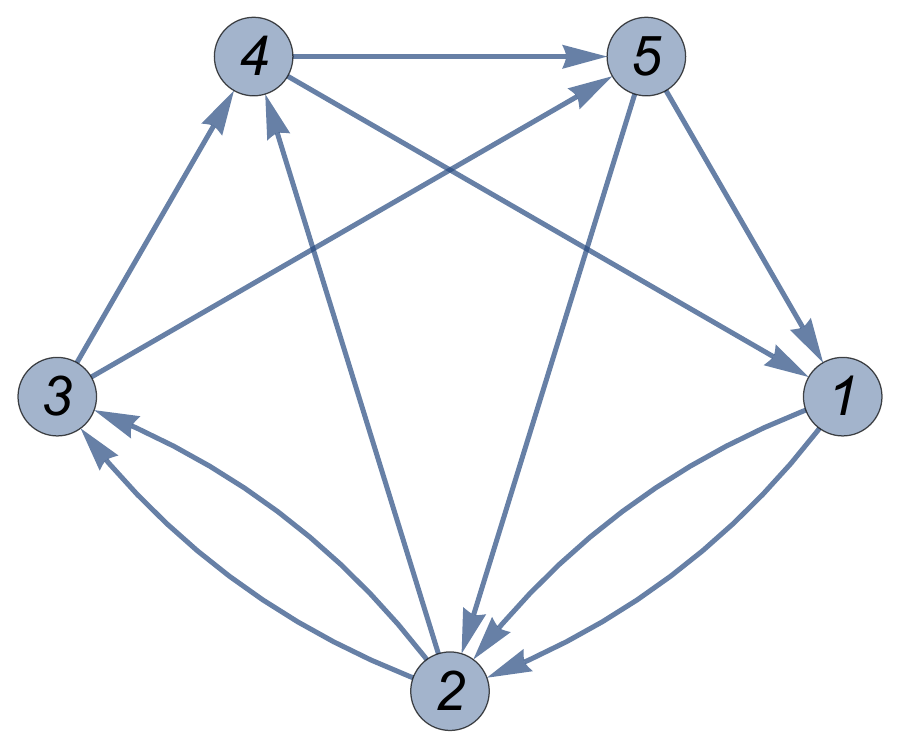}
\end{center}
\caption{BPS quiver for $\text{dP}_2$}
\label{Fig:QuiverdP2}
\end{figure}
The new basis for $\Gamma_f$
\begin{align}
\alpha_0=\gamma_2+\gamma_4+\gamma_5=\alpha_0^{(dP_3)}+\alpha_1^{(dP_3)}, && \alpha_1=\gamma_1+\gamma_3=\alpha_2^{(dP_3)},\nonumber 
\end{align}
\begin{equation}
	\beta_0=\gamma_2-2\gamma_4+2\gamma_5-\gamma_1=2\beta_0^{(dP_3)}+\alpha_0^{(dP_3)}-\alpha_1^{(dP_3)}
\end{equation}
is a basis of simple roots for the lattice $\mathcal{Q}(A_1+A_1)^{(1)}$, as expected from the diagram \ref{Fig:Sakai}, and the intersection form \eqref{eq:CartanD5} in this basis becomes the appropriate Cartan matrix \eqref{eq:CartanA1A1}. The BPS spectrum obtained from the degeneration procedure is
\begin{equation}\label{eq:SpectrumdP2}
\begin{array}{|c|c|}
	\hline
	\gamma & \Omega(\gamma;y) \\
	\hline\hline
	\gamma_r + k v_1 & 1\\
	-\gamma_r + (k+1) v_1 & 1\\
	\gamma_3 + k v_2 & 1\\
	-\gamma_3 + (k+1) v_2 & 1\\
	\hline
	\gamma_5 +k\gamma_{D0} & 1 \\
	-\gamma_5 +(k+1)\gamma_{D0} & 1 \\
	\gamma_1+\gamma_4 +k\gamma_{D0} & 1 \\
	-\gamma_1-\gamma_4 +(k+1)\gamma_{D0} & 1 \\
	v_1+k\gamma_{D0} & y+y^{-1} \\
	-v_1+(k+1)\gamma_{D0} & y+y^{-1} \\
	(k+1)\gamma_{D0} & y^3 + 4y+y^{-1}\\
	\hline
\end{array}
\end{equation}
with $k\geq  0$ and
\begin{align}
r\in\{1,4\}, && v_1=\gamma_1+\gamma_4+\gamma_5, && v_2=\gamma_2+\gamma_3,
\end{align}
which indeed corresponds to the stability condition \eqref{eq:CC1defdP2} and is obtained from the tilting sequence
\begin{equation}
\textbf{m}^{(dP_2)}=\mu_1\mu_3\mu_5\mu_2\mu_4\mu_1\mu_3\mu_5\mu_4\mu_1\mu_2\mu_5\mu_3\mu_4\mu_1\mu_2\mu_5\mu_4,
\end{equation}
corresponding to the affine translation
\begin{equation}
T\left(\vec{\alpha},\vec{\beta} \right)=\left(\alpha_0+6\delta,\alpha_1-6\delta,\beta_0+6\delta,\beta_1-6\delta \right).
\end{equation}
Interestingly, and differently from the previous cases, it does not seem that the mutation sequence can be decomposed into more ``fundamental''  translations. Finally, as we observed in the case of local $dP_4$, one does not have a semiclassically exact solution to the TBAs, which in general will be instead given by solutions of q-Painlev\'e $III_2$ (symmetry type $E_2^{(1)}$ in Sakai's classification). This again corresponds to the fact that the corresponding q-Painlev\'e equation does not have algebraic solutions invariant under foldings \cite{Bershtein:2021gkq}.

\subsection{Local $\mathbb{P}^1\times\mathbb{P}^1$}

Starting from local $dP_2$, there are two possible degenerations, one leading to local $\mathbb{P}^1\times\mathbb{P}^1$, and the other to local $dP_1$. To degenerate to local $\mathbb{P}^1\times\mathbb{P}^1$ we send $Z_4,Z_5\rightarrow\infty$ while keeping $Z_4+Z_5$ finite, i.e. end $\Lambda_2\rightarrow-\infty$, $\tilde{\Lambda}_2\rightarrow+\infty$ with $\Lambda_2+\tilde{\Lambda}_2:=\Lambda$ finite.
After relabeling
\begin{align}
\gamma_4^{(\mathbb{P}^1\times\mathbb{P}^1)}:=\gamma_4^{(dP_2)}+\gamma_5^{(dP_2)}, 
\end{align}
we obtain the quiver in Figure \ref{Fig:QuiverP1P1}, which is the BPS quiver of local $\mathbb{P}^1\times\mathbb{P}^1$, and the stability condition
\begin{align}\label{eq:CC1defP1P1}
\CC_1^{(def)}(\mathbb{P}^1\times\mathbb{P}^1): && Z_1=Z+2\Re Z+\tilde{\Lambda}, && Z_2=\bar{Z}+\Lambda_4, \\
&& Z_3=Z+2\Re Z+\tilde{\Lambda}-\Lambda_8, && Z_4=\bar{Z}+\Lambda_4+\Lambda_8,
\end{align}
which is nothing but the collimation chamber for local $\mathbb{P}^1\times\mathbb{P}^1$ from \cite{DelMonte2022a}, parametrized in a slightly different way.
The new basis for $\Gamma_f$ 
\begin{align}
\alpha_0^{(\mathbb{P}^1\times\mathbb{P}^1)}:=\alpha_1^{(dP_2)} && \alpha_1^{(\mathbb{P}^1\times\mathbb{P}^1)}:=\alpha_0^{(dP_2)}
\end{align}
gives the simple roots of $\mathcal{Q}(A_1^{(1)})$ with intersection matrix \eqref{eq:CartanA1dP1}. The limiting BPS spectrum is
\begin{figure}[h]
\begin{center}
\begin{subfigure}{.5\textwidth}
\centering
\includegraphics[width=.6\textwidth]{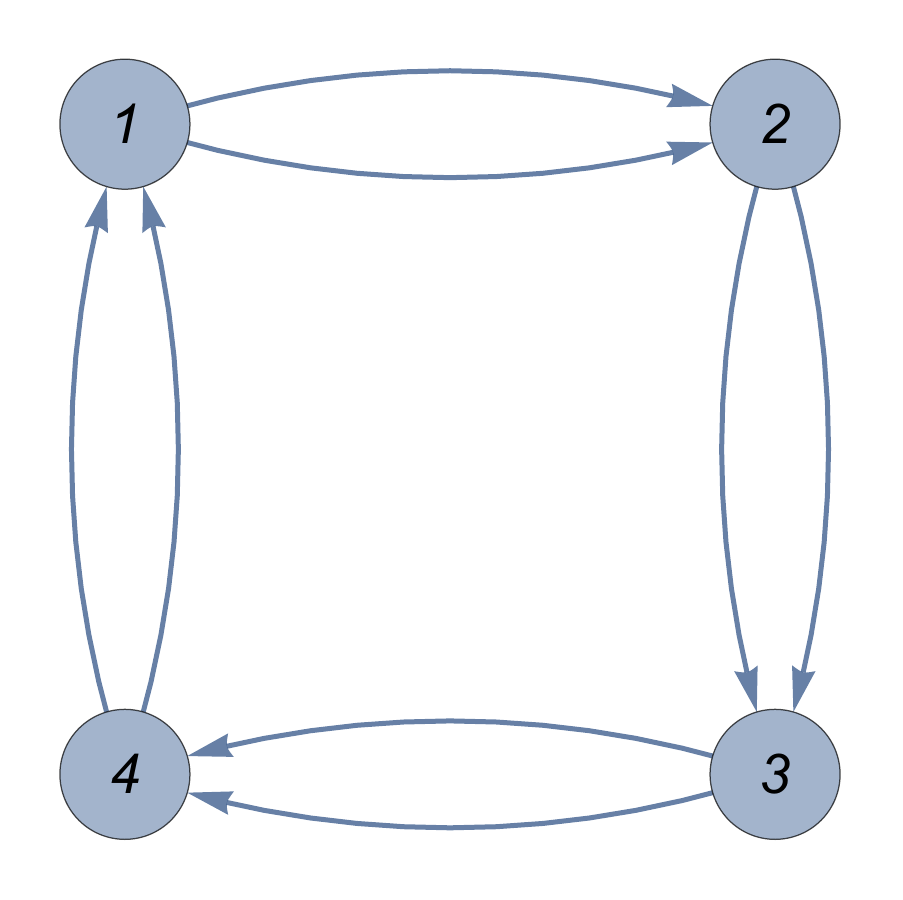}
\caption{BPS quiver for $\mathbb{P}^1\times\mathbb{P}^1$}
\label{Fig:QuiverP1P1}
\end{subfigure}\hfill
\begin{subfigure}{.5\textwidth}
\centering
\includegraphics[width=.6\textwidth]{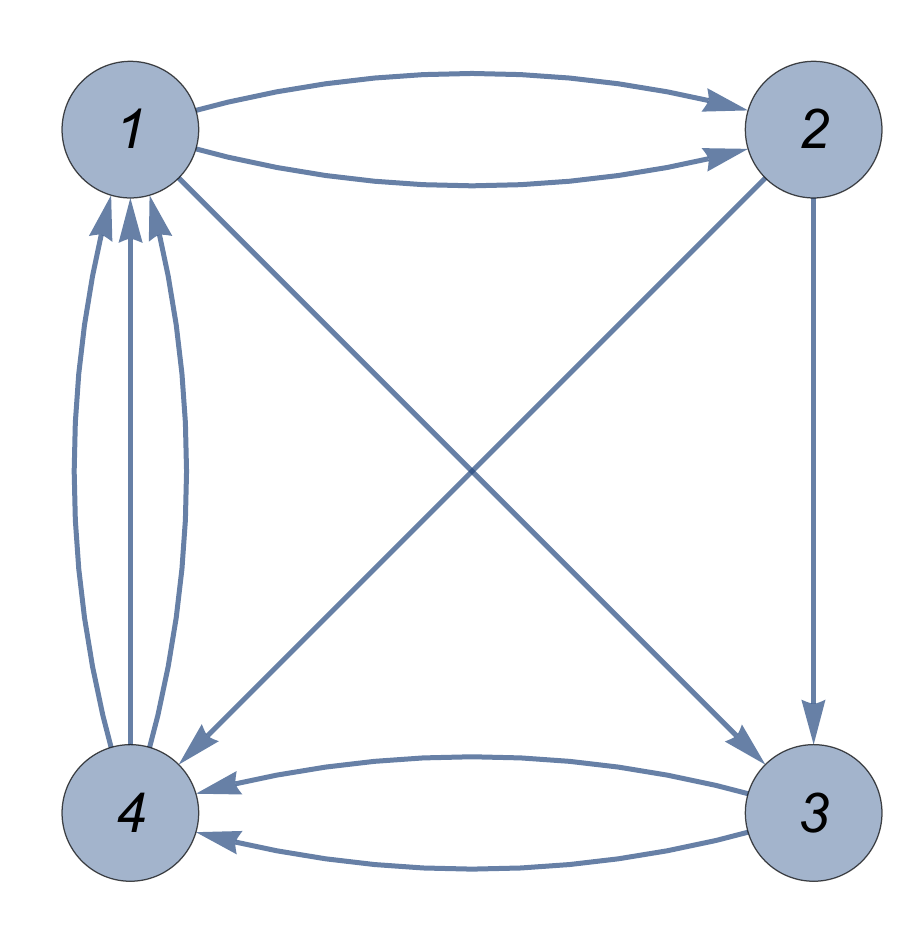}
\caption{BPS quiver for $\text{dP}_1$}
\label{Fig:QuiverdP1}
\end{subfigure}
\end{center}
\end{figure}
\begin{equation}\label{eq:SpectrumP1P1}
\begin{array}{|c|c|}
	\hline
	\gamma & \Omega(\gamma;y) \\
	\hline\hline
	\gamma_1 + k v_1 & 1\\
	-\gamma_1 + (k+1) v_1 & 1\\
	\gamma_3 + k v_2 & 1\\
	-\gamma_3 + (k+1) v_2 & 1\\
	\hline
	v_1+k\gamma_{D0} & y+y^{-1} \\
	-v_1+(k+1)\gamma_{D0} & y+y^{-1} \\
	(k+1)\gamma_{D0} & y^3 + 4y+y^{-1}\\
	\hline
\end{array}
\end{equation}
with $k\geq  0$ and
\begin{align}
v_1=\gamma_1+\gamma_2, && v_2=\gamma_3+\gamma_4,
\end{align}
and coincides with the BPS spectrum for the collimation chamber of local $\mathbb{P}^1\times\mathbb{P}^1$ computed in \cite{Longhi:2021qvz,Bonelli2020,DelMonte2021}.
\subsection{Local $dP_1$}
To degenerate to local $dP_1$ we instead send $Z_1,Z_5\rightarrow\infty$ while keeping $Z_1+Z_5$ finite in \eqref{eq:CC1defdP2}. This limit is slightly less straightforward than the previous ones, as we have to send $\Lambda_4\rightarrow-\infty$, $\tilde{\Lambda}_2,\,\Lambda_8\rightarrow+\infty$, while keeping finite the combinations $
\tilde{\Lambda}_2+\Lambda_4:=\tilde{\Lambda}_3$, $\Lambda_4+\Lambda_8:=\tilde{\Lambda}_4$.
After relabeling
\begin{align}
\gamma_1^{(dP_1)}:=\gamma_2^{(dP_2)}, && \gamma_2^{(dP_1)}:=\gamma_3^{(dP_2)}, && \gamma_3^{(dP_1)}:=\gamma_4^{(dP_2)}, && \gamma_4^{(dP_1)}:=\gamma_1^{(dP_2)}+\gamma_5^{(dP_2)},
\end{align}
we obtain the local $dP_1$ BPS quiver in Figure \ref{Fig:QuiverdP1}, and the stability condition
\begin{align}\label{eq:CC1defdP1}
\CC_1^{(def)}(dP_1): && Z_1=Z+2\Re Z+\Lambda_2+\tilde{\Lambda}_3-\tilde{\Lambda}_4, && Z_2=\bar{Z}+\tilde{\Lambda}_3, \\
&& Z_3=Z+\Lambda_2, && Z_4=\bar{Z}+2\Re Z+\tilde{\Lambda}_3,
\end{align}
with $4\Re Z+\Lambda_2+\tilde{\Lambda}_3=\frac{\pi}{R}$.
The new basis of $\Gamma_f$ is
\begin{align}
\alpha_0^{(dP_1)}:=\frac{1}{2}\alpha_0^{(dP_2)}-\alpha_1^{(dP_2)}+\frac{1}{2}\beta_0^{(dP_2)}, && \alpha_1^{(dP_1)}:=\frac{1}{2}\alpha_0^{(dP_2)}+2\alpha_1^{(dP_2)}-\frac{1}{2}\beta_0^{(dP_2)},
\end{align}
and involves fractional linear combinations of the local $dP_2$ roots. The limiting spectrum is still given by \eqref{eq:SpectrumP1P1}, providing a derivation from decoupling of a proposal made in \cite{Bonelli2020}, where it was conjectured that the BPS spectrum of local $\mathbb{P}^1\times\mathbb{P}^1$ and local $dP_1$ should coincide.

\subsection{Local $\mathbb{P}^2$}

Sadly, our degeneration journey must end here. In order to decouple \eqref{eq:CC1defdP1} to local $\mathbb{P}^2$, we should send $Z_2,Z_3\rightarrow\infty$ while keeping $Z_2+Z_3,Z_1,Z_4$ all finite, and this cannot be done whithin the collimation chamber \eqref{eq:CC1defdP1}. This impossibility is likely due to the fact that the root lattice for this geometry is $A_0^{(1)}$, with only the null root corresponding to the $D0$ brane. As a result, there is no nontrivial affine translation that could correspond to a tilting, which physically corresponds to the absence of a low-energy gauge theory phase for this geometry.

\section{Wild algebraic solutions: the case of local $dP_3$}\label{sec:dP3AD}

Up to now we discussed BPS spectra obtained by either deforming or degenerating the local $dP_5$ stability condition \eqref{eq:dP5algsol} associated to an algebraic solution of q-Painlev\'e VI. Importantly, all the algebraic solutions discussed up to now (including local $dP_3$ and local $\mathbb{P}^1\times\mathbb{P}^1$) had an associated stability condition lying away from a wall of marginal stability. In this section, we will discuss an example where an algebraic solution exists but lies on a wall of marginal stability: as it turns out, it still signals the presence of an interesting chamber, as long as one deforms away from the marginally stable collimated configuration. 

Consider local $dP_3$: the algebraic solution \eqref{eq:algdP3} was obtained from the configuration invariant under the $\mathbb{Z}_2$ involution $\pi^3=(1,4)(2,5)(3,6)$, and evolved according to the time evolution \eqref{eq:T1dP3} on the $A_2^{(1)}$ sublattice of $\Gamma_f\simeq\mathcal{Q}(A_2+A_1)^{(1)}$. We can also consider the algebraic solution
\begin{align}\label{eq:dP3A2alg}
\CX_1=\CX_3=\CX_5=b_0^{-1/3}, && \CX_2=\CX_4=\CX_6=qb_0^{1/3},
\end{align}
invariant under the $\mathbb{Z}_3$ quiver automorphism $\pi^2=(1,3,5)(2,4,6)$. It evolves under the (q-Painlev\'e IV) time evolution 
\begin{equation}
	T_4=\mu_2\mu_6\mu_4\mu_2(1,4,3,6)(2,5)=(1,4,3,6)(2,5)\mu_5\mu_3\mu_1\mu_5,
\end{equation}
acting on the $A_1^{(1)}$ sublattice generated by the roots $\beta_0,\beta_1$ in \eqref{eq:dP3-beta-roots}. The action of $T_4$ on the BPS charges is
\begin{align}
	T_4^{3n+1}(\vec{\gamma})=\left( \begin{array}{c}
	\gamma_1+\gamma_3+\gamma_4+n\delta \\
	-\gamma_1-n\delta \\
	\gamma_3+\gamma_5+\gamma_6+n\delta \\
	-\gamma_3-n\delta \\
	\gamma_5+\gamma_1+\gamma_2+n\delta \\
	-\gamma_5-n\delta 		
	\end{array} \right), &&
	T_4^{3n+2}(\vec{\gamma})=\left( \begin{array}{c}
	\gamma_2-(n-1)\delta \\
	-(\gamma_1+\gamma_3+\gamma_4)-n\delta \\
	\gamma_4+(n-1)\delta \\
	-(\gamma_3+\gamma_5+\gamma_6)+n\delta \\
	\gamma_6-(n-1)\delta \\
	-(\gamma_1+\gamma_2+\gamma_5)-n\delta		
	\end{array} \right),
\end{align}
\begin{equation}
	T_4^{3n}(\vec{\gamma})=\left( \begin{array}{c}
	\gamma_1+n\delta \\
	\gamma_2-n\delta \\
	\gamma_3+n\delta \\
	\gamma_4-n\delta \\
	\gamma_5+n\delta \\
	\gamma_6-n\delta		
	\end{array} \right),
\end{equation}
while on the (multiplicative) root variables it acts as
\begin{align}
	T_4(\alpha_0,\alpha_1,\alpha_2,\beta_0,\beta_1)=(\alpha_0,\alpha_1,\alpha_2,\beta_0+\delta,\beta_1-\delta), \\ T_4(a_0,a_1,a_2,b_0,b_1)=(a_0,a_1,a_2,qb_0,q^{-1}b_1).
\end{align}
The corresponding stability condition
\begin{equation}\label{eq:dP3MargStab}
\CC_1^{(alg)}(dP_3)':\qquad Z:=Z_1=Z_3=Z_5,\qquad Z_2=Z_3=Z_6=\frac{2\pi}{3R}-Z
\end{equation}

lies on a wall of marginal stability, because mutually nonlocal central charges are aligned. Let us ignore this for a moment and work with it as a formal stability condition. First note that, up to permutations, if we define
\begin{align}
\textbf{m}_1:=\mu_1\mu_5\mu_1\mu_3\mu_1 && \textbf{m}_2:=\mu_2\mu_6\mu_2\mu_4\mu_2,
\end{align}
we have
\begin{align}
	\textbf{m}_1\sim T_4, && \textbf{m}_2\textbf{m}_1\sim T_4^2, && \textbf{m}_1\textbf{m}_2\textbf{m}_1\sim T_4^3,
\end{align}
i.e. the mutation sequence $\textbf{m}_1\textbf{m}_2\textbf{m}_1$ is induced by the affine translation $T_4$ \footnote{It might not seem so, because $T_4$ has 4 mutation while $\textbf{m}_i$ have 5. However there are relations between the mutations because $135$ and $246$ are closed cycles in the quiver.}. Keeping track of the states that are being acted upon by mutations (in red below), we have
\begin{equation}
\begin{split}
	\textbf{m}_1&:\left( \begin{array}{c}
	{\color{red}\gamma_1} \\
	\gamma_2 \\
	\gamma_3 \\
	\gamma_4 \\
	\gamma_5 \\
	\gamma_6		
	\end{array} \right)  
	\xrightarrow{\mu_1}\left( \begin{array}{c}
	-\gamma_1 \\
	\gamma_1+\gamma_2 \\
	{\color{red}\gamma_1+\gamma_3} \\
	\gamma_4 \\
	\gamma_5 \\
	\gamma_6		
	\end{array} \right)
	\xrightarrow{\mu_3}\left( \begin{array}{c}
	{\color{red}\gamma_3} \\
	\gamma_1+\gamma_2 \\
	-\gamma_1-\gamma_3 \\
	\gamma_1+\gamma_3+\gamma_4 \\
	\gamma_5 \\
	\gamma_6		
	\end{array} \right)
	\xrightarrow{\mu_1}\left( \begin{array}{c}
	-\gamma_3 \\
	\gamma_1+\gamma_2 \\
	-\gamma_1 \\
	\gamma_1+\gamma_3+\gamma_4 \\
	{\color{red}\gamma_3+\gamma_5} \\
	\gamma_6		
	\end{array} \right) \\
	&\xrightarrow{\mu_5}\left( \begin{array}{c}
	{\color{red}\gamma_5} \\
	\gamma_1+\gamma_2 \\
	-\gamma_1 \\
	\gamma_1+\gamma_3+\gamma_4 \\
	-\gamma_3-\gamma_5 \\
	\gamma_3+\gamma_5+\gamma_6		
	\end{array} \right)\xrightarrow{\mu_1}\left( \begin{array}{c}
	-\gamma_5 \\
	\gamma_1+\gamma_2+\gamma_5 \\
	-\gamma_1 \\
	\gamma_1+\gamma_3+\gamma_4 \\
	-\gamma_3 \\
	\gamma_3+\gamma_5+\gamma_6		
	\end{array} \right)=(1,4,5,2,3,6)T_4(\vec{\gamma}),
\end{split}
\end{equation}
\begin{equation}
\begin{split}
	&\textbf{m}_2:\left( \begin{array}{c}
	-\gamma_5 \\
	\color{red}\gamma_1+\gamma_2+\gamma_5 \\
	-\gamma_1 \\
	\gamma_1+\gamma_3+\gamma_4 \\
	-\gamma_3 \\
	\gamma_3+\gamma_5+\gamma_6		
	\end{array} \right)
	\xrightarrow{\mu_2}\left( \begin{array}{c}
	\gamma_1+\gamma_2 \\
	-\gamma_1-\gamma_2-\gamma_5 \\
	-\gamma_1\\
	\color{red}2\gamma_1+\gamma_2+\gamma_3+\gamma_4+\gamma_5 \\
	-\gamma_3 \\
	\gamma_3+\gamma_5+\gamma_6		
	\end{array} \right)\\
	&\xrightarrow{\mu_4}\left( \begin{array}{c}
	\gamma_1+\gamma_2 \\
	\color{red}\gamma_1+\gamma_3+\gamma_4 \\
	\gamma_1+\gamma_2+\gamma_3+\gamma_4+\gamma_5\\
	-2\gamma_1-\gamma_2-\gamma_3-\gamma_4-\gamma_5 \\
	-\gamma_3 \\
	\gamma_3+\gamma_5+\gamma_6		
	\end{array} \right)
	\xrightarrow{\mu_2}\left( \begin{array}{c}
	\gamma_1+\gamma_2 \\
	-\gamma_1-\gamma_3-\gamma_4 \\
	\gamma_1+\gamma_2+\gamma_3+\gamma_4+\gamma_5\\
	-\gamma_1-\gamma_2-\gamma_5 \\
	-\gamma_3 \\
	\color{red}\gamma_1+2\gamma_3+\gamma_4+\gamma_5+\gamma_6		
	\end{array} \right)\\
	&\xrightarrow{\mu_6}\left( \begin{array}{c}
	\gamma_1+\gamma_2 \\
	\color{red}\gamma_3+\gamma_5+\gamma_6 \\
	\gamma_1+\gamma_2+\gamma_3+\gamma_4+\gamma_5\\
	-\gamma_1-\gamma_2-\gamma_5 \\
	\gamma_1+\gamma_3+\gamma_4+\gamma_5+\gamma_6 \\
	-\gamma_1-2\gamma_3-\gamma_4-\gamma_5-\gamma_6		
	\end{array} \right)
	\xrightarrow{\mu_2}\left( \begin{array}{c}
	\gamma_1+\gamma_2+\gamma_3+\gamma_5+\gamma_6 \\
	-\gamma_3-\gamma_5-\gamma_6 \\
	\gamma_1+\gamma_2+\gamma_3+\gamma_4+\gamma_5\\
	-\gamma_1-\gamma_2-\gamma_5 \\
	\gamma_1+\gamma_3+\gamma_4+\gamma_5+\gamma_6 \\
	-\gamma_1-\gamma_3-\gamma_4		
	\end{array} \right)\\
	&=(1,5,3)(2,6,4)T_4^2(\vec{\gamma})
\end{split}
\end{equation}
\begin{equation}
\begin{split}
&\textbf{m}_1:\left( \begin{array}{c}
	\color{red}\gamma_1+\gamma_2+\gamma_3+\gamma_5+\gamma_6 \\
	-\gamma_3-\gamma_5-\gamma_6 \\
	\gamma_1+\gamma_2+\gamma_3+\gamma_4+\gamma_5\\
	-\gamma_1-\gamma_2-\gamma_5 \\
	\gamma_1+\gamma_3+\gamma_4+\gamma_5+\gamma_6 \\
	-\gamma_1-\gamma_3-\gamma_4		
	\end{array} \right)\xrightarrow{\mu_1}\left( \begin{array}{c}
	-\gamma_1-\gamma_2-\gamma_3-\gamma_5-\gamma_6 \\
	\gamma_1+\gamma_2 \\
	\color{red}2\gamma_1+2\gamma_2+2\gamma_3+\gamma_4+2\gamma_5+\gamma_6\\
	-\gamma_1-\gamma_2-\gamma_5 \\
	\gamma_1+\gamma_3+\gamma_4+\gamma_5+\gamma_6 \\
	-\gamma_1-\gamma_3-\gamma_4		
	\end{array} \right)\\
	&\xrightarrow{\mu_3}\left( \begin{array}{c}
	\color{red}\gamma_1+\gamma_2+\gamma_3+\gamma_4+\gamma_5 \\
	\gamma_1+\gamma_2 \\
	-2\gamma_1-2\gamma_2-2\gamma_3-\gamma_4-2\gamma_5-\gamma_6\\
	\gamma_1+\gamma_2+2\gamma_3+\gamma_4+\gamma_5+\gamma_6 \\
	\gamma_1+\gamma_3+\gamma_4+\gamma_5+\gamma_6 \\
	-\gamma_1-\gamma_3-\gamma_4		
	\end{array} \right)
	\xrightarrow{\mu_1}\left( \begin{array}{c}
	-\gamma_1-\gamma_2-\gamma_3-\gamma_4-\gamma_5 \\
	\gamma_1+\gamma_2 \\
	-\gamma_1-\gamma_2-\gamma_3-\gamma_5-\gamma_6\\
	\gamma_1+\gamma_2+2\gamma_3+\gamma_4+\gamma_5+\gamma_6 \\
	\color{red}2\gamma_1+\gamma_2+2\gamma_3+2\gamma_4+2\gamma_5+\gamma_6 \\
	-\gamma_1-\gamma_3-\gamma_4		
	\end{array} \right)\\
	&\xrightarrow{\mu_5}\left( \begin{array}{c}
	\color{red}\gamma_1+\gamma_3+\gamma_4+\gamma_5+\gamma_6 \\
	\gamma_1+\gamma_2 \\
	-\gamma_1-\gamma_2-\gamma_3-\gamma_5-\gamma_6\\
	\gamma_1+\gamma_2+2\gamma_3+\gamma_4+\gamma_5+\gamma_6 \\
	-2\gamma_1-\gamma_2-2\gamma_3-2\gamma_4-2\gamma_5-\gamma_6 \\
	\gamma_1+\gamma_2+\gamma_3+\gamma_4+2\gamma_5+\gamma_6		
	\end{array} \right)
	\xrightarrow{\mu_1}\left( \begin{array}{c}
	\gamma_2-\gamma_{D0} \\
	\gamma_1+\gamma_{D0} \\
	\gamma_4-\gamma_{D0}\\
	\gamma_3+\gamma_{D0} \\
	\gamma_6-\gamma_{D0} \\
	\gamma_5+\gamma_{D0}		
	\end{array} \right)=(1,2)(3,4)(5,6)T_4^3(\vec{\gamma}).
\end{split}
\end{equation}
Note that this is also the mutation sequence that would be induced by a tilting of the upper half-plane for the stability condition \eqref{eq:dP3MargStab}, but since we are on a wall of marginal stability the ordering of mutations in a tilting is ambiguously defined. The solution is to slightly deform the stability condition, into
\begin{align}\label{eq:CC1defdP3p}
	\CC_1^{(def)}(dP_3)':\qquad Z_1=Z-\Lambda_1, \quad Z_2=\bar{Z}-\Lambda_2, \quad Z_3=Z, \\
	 Z_4=\bar{Z}+\Lambda_1, \quad Z_5=Z+\Lambda_2, \quad Z_6=\bar{Z}.
\end{align}
For small enough positive $\Lambda_1,\,\Lambda_2$ the phase ordering of the central charges on the upper-right quadrant is given by the states in red in the equations above, so that $T_4$ indeed represents a tilting for this stability condition. The states in the lower-right quadrand can be obtained in a similar way by performing the tilting counterclockwise. 
Similarly to what we saw in Section \ref{sec:degens}, we have Peacock patterns of states with a single limiting ray, the real axis, and for generic $\Lambda_1,\,\Lambda_2$ no central charges of mutually nonlocal states are aligned. 

While the states outside the limiting ray are those produced by $T_4$, the ones on the real axis can be found following \cite{DelMonte2021} using the permutations symmetries of the quiver, in the following way. First recall that the Kontsevich-Soibelman wall-crossing invariant/quantum monodromy $\mathbb{U}$ is defined as:
\begin{equation}\label{eq:factorization}
\mathbb{U}=\prod_{\gamma\in\Gamma_+}^{\curvearrowleft}\prod_{m\in \IZ} \Phi((-y)^m X_{\gamma})^{a_m(\gamma)}.
\end{equation}
Here 
\begin{itemize}
\item $\Gamma_+:=\{\gamma\in\Gamma:\,\Re Z_\gamma>0\}$;
\item The phase ordering $\curvearrowleft$ is defined by decreasing $\arg Z_\gamma\in\left(\frac{\pi}{2},-\frac{\pi}{2} \right)$ from left to right;
\item $\Omega(\gamma,y):=\sum_m(-y)^ma_m(\gamma)$ is the Protected Spin Character \cite{Gaiotto:2010be}, coinciding with the BPS index of \cite{Gaiotto:2009hg} when $y=-1$. Geometrically they are respectively the motivic and unrefined Donaldson-Thomas invariants of the coherent sheaf with Chern character vector $\gamma$;
\item The quantum dilogarithm is $\Phi(x):=\prod_{n=0}^\infty(1+y^{2n+1}x)^{-1} $;
\item $X_\gamma$ is a noncommutative deformation of the $\CX$-cluster variables, such that $X_\gamma X_{\gamma'}=y^{\langle\gamma,\gamma'\rangle}X_{\gamma+\gamma'}$.
\end{itemize}
In a given chamber $\CC$, the wall-crossing invariant $\mathbb{U}$ can be factorized into an ordered product $\mathbb{U}(\measuredangle^+,\CC )$ from states with central charges lying on the upper-right quadrant, $\mathbb{U}_0(\CC)$ from the real axis, and a product on the real axis, and $\mathbb{U}(\measuredangle^-,\CC )$ from the lower-right quadrant:
\begin{equation}\label{eq:UFactor}
\mathbb{U}(\CC_1')=\mathbb{U}(\measuredangle^+,\CC_1')\mathbb{U}_0(\CC_1')\mathbb{U}(\measuredangle^-,\CC_1').
\end{equation}
In the chamber \eqref{eq:CC1defdP3p}, we can immediately write down the contributions from states outside the real axis:
{\scriptsize\begin{equation}
\begin{split}
\mathbb{U}(\measuredangle^{+},\CC_1')&=\prod_{k\ge0}^{\nearrow}\bigg(\Phi\left(X_{\gamma_1+k\gamma_{D0}}\right)\Phi\left(X_{\gamma_1+\gamma_3}+2k\gamma_{D0} \right) \Phi\left(X_{\gamma_3}+k\gamma_{D0} \right)\Phi\left(X_{\gamma_3+\gamma_5}+2k\gamma_{D0} \right)\Phi\left(X_{\gamma_5}+k\gamma_{D0} \right)\bigg) \\
& \times\bigg(\Phi\left(X_{\gamma_1+\gamma_2+\gamma_5}+k\gamma_{D0} \right)\Phi\left(X_{2\gamma_1+\gamma_2+\gamma_3+\gamma_4+\gamma_5}+2k\gamma_{D0} \right)\Phi\left(X_{\gamma_1+\gamma_3+\gamma_4}+k\gamma_{D0} \right)\\
&\times\Phi\left(X_{\gamma_1+2\gamma_3+\gamma_4+\gamma_5+\gamma_6}+2k\gamma_{D0} \right)\Phi\left(X_{\gamma_3+\gamma_5+\gamma_6}+k\gamma_{D0} \right)  \bigg) \\
&\times\bigg(\Phi\left(X_{-\gamma_4+(k+1)\gamma_{D0}}\right)\Phi\left(X_{-(\gamma_4+\gamma_6)}+2(k+1)\gamma_{D0} \right) \Phi\left(X_{-\gamma_6}+(k+1)\gamma_{D0} \right)\\
& \times \Phi\left(X_{-(\gamma_2+\gamma_6)}+2(k+1)\gamma_{D0} \right)\Phi\left(X_{-\gamma_2}+(k+1)\gamma_{D0} \right)\bigg),
\end{split}
\end{equation}}
{\scriptsize\begin{equation}
\begin{split}
\mathbb{U}(\measuredangle^{-},\CC_1')&=\prod_{k\ge0}^{\searrow}\bigg(\Phi\left(X_{\gamma_2+k\gamma_{D0}}\right)\Phi\left(X_{\gamma_2+\gamma_6}+2k\gamma_{D0} \right) \Phi\left(X_{\gamma_6}+k\gamma_{D0} \right)\Phi\left(X_{\gamma_4+\gamma_6}+2k\gamma_{D0} \right)\Phi\left(X_{\gamma_4}+k\gamma_{D0} \right)\bigg) \\
& \times\bigg(\Phi\left(X_{\gamma_1+\gamma_2+\gamma_4}+k\gamma_{D0} \right)\Phi\left(X_{\gamma_1+2\gamma_2+\gamma_4+\gamma_5+\gamma_6}+2k\gamma_{D0} \right)\Phi\left(X_{\gamma_2+\gamma_5+\gamma_6}+k\gamma_{D0} \right)\\
&\times\Phi\left(X_{\gamma_2+\gamma_3+\gamma_4+\gamma_5+2\gamma_6}+2k\gamma_{D0} \right)\Phi\left(X_{\gamma_3+\gamma_4+\gamma_6}+k\gamma_{D0} \right)  \bigg) \\
&\times\bigg(\Phi\left(X_{-\gamma_5+(k+1)\gamma_{D0}}\right)\Phi\left(X_{-(\gamma_3+\gamma_5)}+2(k+1)\gamma_{D0} \right) \Phi\left(X_{-\gamma_3}+(k+1)\gamma_{D0} \right)\\
& \times \Phi\left(X_{-(\gamma_1+\gamma_3)}+2(k+1)\gamma_{D0} \right)\Phi\left(X_{-\gamma_1}+(k+1)\gamma_{D0} \right)\bigg)
\end{split}
\end{equation}}
where $\nearrow$ means increasing $k$ from left to right, while $\searrow$ the opposite. To obtain the states along the real axis, note that under the $\mathbb{Z}_6$ permutation symmetry $\pi=(1,2,3,4,5,6)$ of the quiver \ref{Fig:QuiverdP3}, the stability condition \eqref{eq:CC1defdP3p} gets mapped to other equivalent ones of the same type, with spectrum related by the permutation $\pi$. Denote by $\CC_n:=\pi^{n-1}(\CC_1) $, $n=1,\dots,6$. $\mathbb{U}$ is a wall-crossing invariant, so
\begin{equation}
\mathbb{U}(\CC_i')=\mathbb{U}(\CC_j'),
\end{equation}
which can be recast into an equation for $\mathbb{U}_0$, to be solved order by order in the ``size'' of the BPS state, i.e. the total number of elementary charges of which it is composed. After re-factoring $\mathbb{U}_0$ in quantum dilogarithms, the contributions up to size 8 are
\begin{equation}\label{eq:SpectrumdP3real}
\begin{array}{|c|c|}
	\hline
	\gamma:Z_\gamma\in\mathbb{R}_+ & \Omega(\gamma;y) \\
	\hline\hline
	\gamma_5+\gamma_6 & 1\\
	\gamma_3+\gamma_4 & 1\\
	\gamma_1+\gamma_2 & 1\\
	\gamma_3+\gamma_4+\gamma_5+\gamma_6 & y^{-2}+1+y^2\\
	\gamma_1+\gamma_2+\gamma_3+\gamma_4 & y^{-2}+1+y^2 \\
	\gamma_3+\gamma_4+2(\gamma_5+\gamma_6) & y^{-2}+1+y^2 \\
	\gamma_1+\gamma_2+2(\gamma_3+\gamma_4) & y^{-2}+1+y^2 \\
	2(\gamma_3+\gamma_4)+\gamma_5+\gamma_6 & -\left(y^{-4}+y^{-2}+1+y^2+y^4 \right) \\
	2(\gamma_1+\gamma_2)+\gamma_3+\gamma_4 & -\left(y^{-4}+y^{-2}+1+y^2+y^4 \right) \\
	\vdots & \vdots \\
	\hline
	\hline
\end{array}
\end{equation}
To complete the spectrum, we have to include all the towers of hypermultiplet states:
\begin{equation}\label{eq:SpectrumdP3p}
\begin{array}{|c|c|}
	\hline
	\gamma:\,Z_\gamma\notin\mathbb{R} & \Omega(\gamma;y) \\
	\hline\hline
	\gamma_{1,3,5}+k\gamma_{D0}		& 1\\
  \gamma_{2,4,6}+k\gamma_{D0}   & 1\\
  -\gamma_{1,3,5}+(k+1)\gamma_{D0}    & 1\\
	-\gamma_{2,4,6}+(k+1)\gamma_{D0} 		& 1\\
	\gamma_{1,3}+\gamma_{3,5}+2k\gamma_{D0}		& 1\\
  \gamma_{2,4}+\gamma_{4,6}+2k+\gamma_{D0}    & 1\\
  -(\gamma_{1,3}+\gamma_{3,5})+2(k+1)\gamma_{D0}   & 1\\
	-(\gamma_{2,4}+\gamma_{4,6})+2(k+1)\gamma_{D0} 		& 1\\
	\gamma_{1,1,3}+\gamma_{2,3,5}+\gamma_{5,4,6}+k\gamma_{D0}	& 1 \\
	-(\gamma_{1,1,3}+\gamma_{2,3,5}+\gamma_{5,4,6})+(k+1)\gamma_{D0}	& 1 \\
	2\gamma_1+\gamma_2+\gamma_3+\gamma_4+\gamma_5+2k\gamma_{D0}	& 1 \\
	\gamma_1+2\gamma_3+\gamma_4+\gamma_5+\gamma_6+2k\gamma_{D0} 	& 1 \\
	-(2\gamma_1+\gamma_2+\gamma_3+\gamma_4+\gamma_5)+2(k+1)\gamma_{D0} 	& 1 \\
	-(\gamma_1+2\gamma_3+\gamma_4+\gamma_5+\gamma_6)+2(k+1)\gamma_{D0} 	& 1 \\
	\hline
\end{array}
\end{equation}
where the notation $\gamma_{i,j}+\gamma_{k,l}$ means that both the states $\gamma_i+\gamma_k$ and $\gamma_j+\gamma_l$ are present, and $k\ge0$. Note that, even though we deformed the original chamber, marginally stable states are still present, although these are only the higher-spin states on the real axis. Thus, the stability condition \eqref{eq:CC1defdP3p} still lies on a wall of marginal stability, although in a less evident way. However, since all the marginally stable states lie only on the real axis, the factorization \eqref{eq:UFactor} and the wall-crossing invariant we computed still makes sense, and in fact marginally stable stability conditions have already been used in the four-dimensional setting to efficiently compute the wall-crossing invariant $\mathbb{U}$ \cite{Longhi:2016wtv}.

Let us conclude this section with the following remark: it was proposed in \cite{Bonelli2020} that the BPS spectrum produced by the time evolution $T_4$ is naturally associated to a 5d uplift of the nonlagrangian Argyres-Douglas $(A_1,A_3)$ theory: upon decoupling the nodes $3$ and $6$ from the local $dP_3$ quiver, one obtains the appropriate four-dimensional BPS quiver, see Figure \ref{Fig:dP3T4}. In this respect, it is worth noting that in the q-Painlev\'e confluence diagram in Figure \ref{Fig:Sakai} there are also arrows from local $\mathbb{P}^1\times\mathbb{P}^1$ to the $(A_1,A_2)$ Argyres-Douglas theory ($A_1^{(1)}$ in the bottom row), and from local $\mathbb{P}^2$ to the $(A_1,A_1)$ theory ($A_0^{(1)}$ in the bottom row). 
\begin{figure}[h]
\begin{center}
\begin{subfigure}{.45\textwidth}
\centering
\includegraphics[width=.6\textwidth]{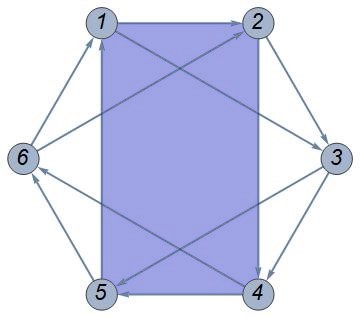}
\caption{Argyres Douglas subquiver of the local $dP_3$ subquiver}\label{Fig:dP3T4}
\end{subfigure}\hfill
\begin{subfigure}{.45\textwidth}
\centering
\includegraphics[width=.6\textwidth]{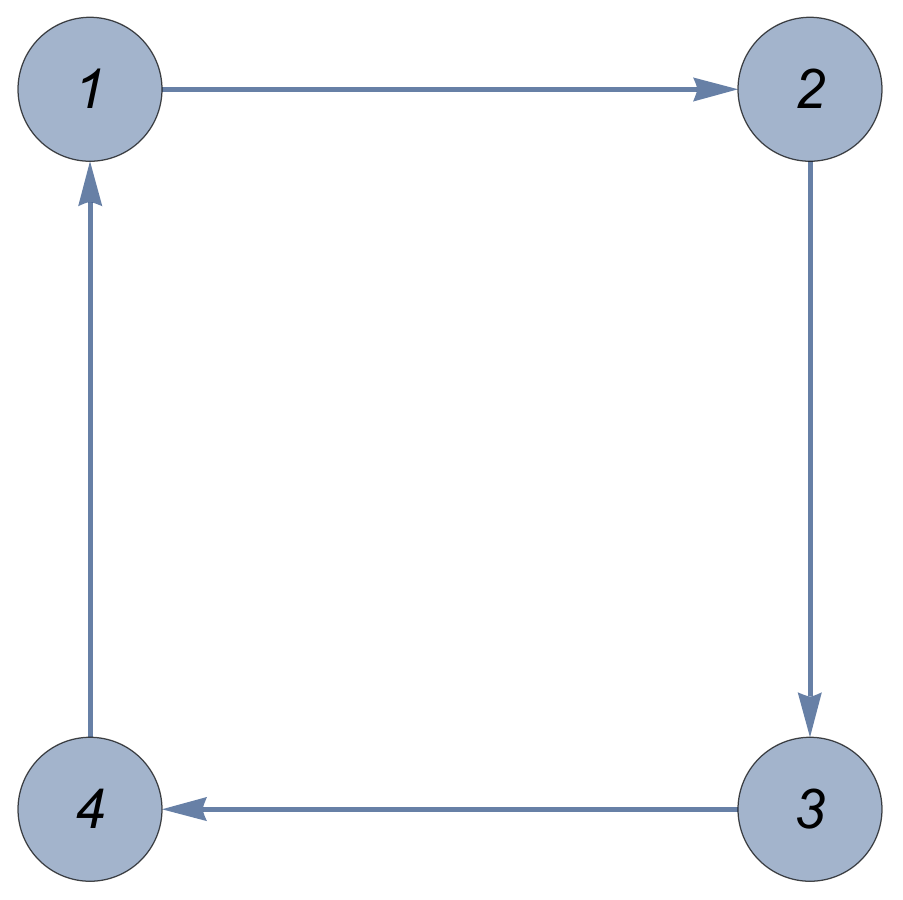}
\caption{BPS quiver for the four-dimensional $(A_1,D_4)$ Argyres-Douglas theory}
\label{Fig:QuiverD4}
\end{subfigure}
\end{center}
\end{figure}

The analysis of this section shows that such an uplift should display a wild spectrum, as opposed to the finite spectrum of the four-dimensional theory. Indeed, to fully go away from the wall of marginal stability we must further deform \eqref{eq:CC1defdP3p} in such a way that $Z_1,Z_3,Z_5$ no longer have the same imaginary part, and similarly for $Z_2,Z_4,Z_6$. This takes us away from the collimation region, as it opens up a cone where all the higher-spin states will lie.

\section{Conclusions and outlook}\label{sec:Concl}

The existence of algebraic solutions of the cluster IS proved to be a powerful tool to investigate BPS spectra: it allowed to find collimation chambers, both of tame and wild type (albeit this latter only seem to signal a marginally stable region near a properly wild chamber), and compute the corresponding BPS spectra and solutions to the TBA equations. BPS spectra of geometries that do not admit algebraic solutions can be obtained from those that do by an appropriate decoupling procedure. Having explored the correspondence between BPS spectra and algebraic solutions of cluster integrable systems for the five-dimensional theories with low-energy $SU(2)$ gauge theory phases with fundamental matter, there are still many directions to be explored.

\subsubsection*{Foldings of Cluster Integrable Systems and collimation chambers}

In this work we studied algebraic solutions of cluster integrable systems characterized by invariance under appropriate permutations symmetries of the corresponding quiver, and showed that they can be identified with solutions of the conformal TBAs associated to the BPS spectrum. However, there can be more general algebraic solutions invariant under foldings, i.e. cluster automorphisms involving also mutations \cite{Bershtein:2021gkq}. Do these algebraic solutions have a counterpart from the point of view of the BPS spectral problem? If so, the discussion of this paper would need to be quite nontrivially generalized, as (for example) a cluster automorphism involving mutations necessarily does not preserve a choice of positive half-plane, so that we cannot think in terms of a single stability condition anymore.

Another intriguing point is the following: folding transformations in general send solutions of one q-Painlev\'e equation to solutions of a different one, so that algebraic solutions solve more than one q-Painlev\'e equation, after appropriate redefinitions\footnote{The author thanks M. Bershtein for drawing his attention to this point.}. This could point to some new duality between different geometries, where the BPS spectra of two different theories could be identified in appropriate regions of moduli space, and is certainly a question worth investigating.

\subsubsection*{Decoupling limits and four-dimensional theories}

In Section \ref{sec:degens} we studied degeneration limits of stability conditions corresponding to the decoupling of a fundamental hypermultiplet in the low-energy gauge theory phase. Another natural limit from the physical point of view is the four-dimentional one, connecting the five-dimensional theory with the corresponding four-dimensional gauge theory. In the four dimensional setting, on the one hand there is the differential Painlev\'e equation, whose solution is an appropriate limit of the q-Painlev\'e solutions and whose connection to the four-dimensional gauge theory is well-understood through the Painlev\'e-gauge theory correspondence; on the other hand, the cluster variables appear in this setting as coordinates on the monodromy variety corresponding to the Painlev\'e equation \cite{chekhov2017painleve}. In some sense, in the five-dimensional setting monodromy and isomonodromy seem to be unified in a mysterious way begging for an explanation.

One could also consider some physically unexpected degenerations, that might be indicating the existence of (to the author's knowledge) unstudied limits from five-dimensional SCFTs to four-dimensional $\mathcal{N}=2$ theories. For example, we could continue the degeneration scheme of Section \ref{sec:degens}, starting from the collimation chamber stability condition \eqref{eq:CC1defP1P1}, that we write as
\begin{equation}
\CC_1^{(def)}(\mathbb{P}^1\times\mathbb{P}^1):\qquad Z_1=Z, \quad Z_3=Z-\Lambda_1, \quad Z_2=\bar{Z}+\Lambda_2, \quad Z_4=\bar{Z}+\Lambda_1+\Lambda_2,
\end{equation}
and take the limit $\Lambda_1\rightarrow\infty$. This decouples the nodes $\gamma_3,\,\gamma_4$ while keeping finite $\gamma_3+\gamma_4$. The resulting quiver is the 2-Markov quiver in Figure \ref{Fig:QuiverN2star},
\begin{figure}[h]
\begin{center}
\includegraphics[width=.35\textwidth]{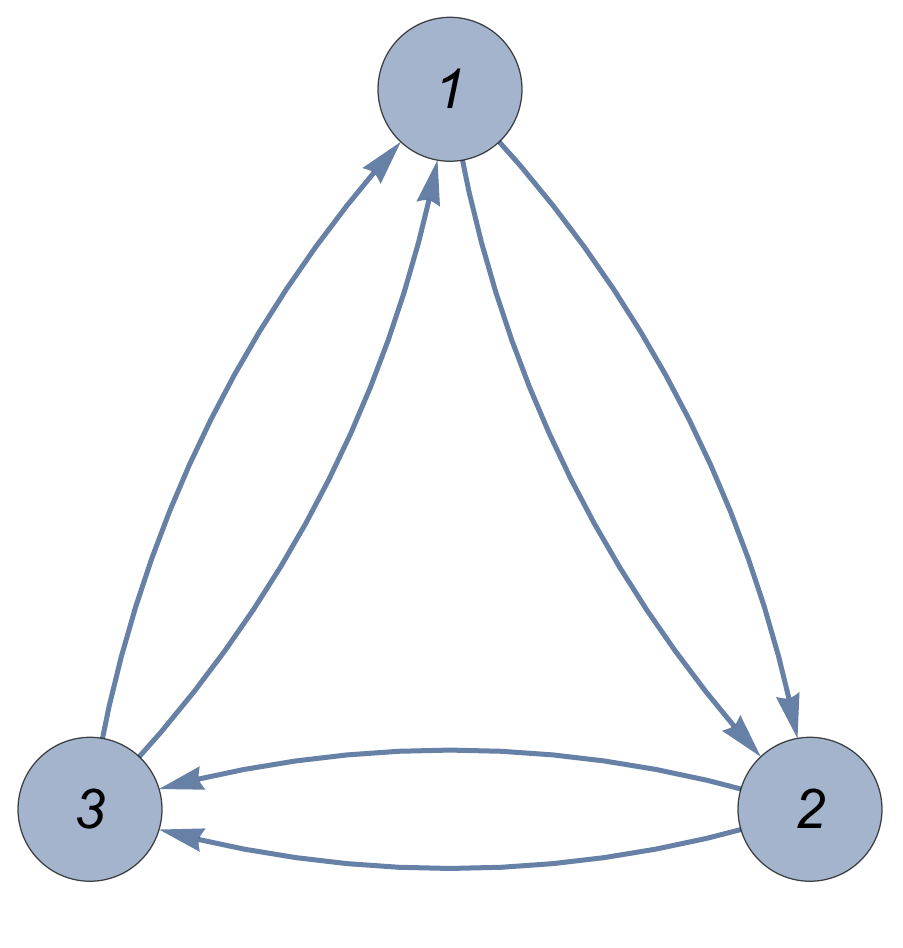}
\end{center}
\caption{BPS quiver for four-dimensional $\mathcal{N}=2^*$ theory}
\label{Fig:QuiverN2star}
\end{figure}
which is the BPS quiver of the $\mathcal{N}=2^*$ theory \cite{Alim:2011ae}, i.e. four-dimensional $SU(2)$ SYM with an adjoint hypermultiplet. It is the quiver associated to the character variety of the one-punctured torus $C_{1,1}$. After relabeling $\gamma_3+\gamma_4\mapsto\gamma_3$, the local $\mathbb{P}^1\times\mathbb{P}^1$ spectrum \eqref{eq:SpectrumP1P1} in this limit becomes the following:
\begin{equation}\label{eq:SpectrumMarkov}
\begin{array}{|c|c|}
  \hline
  \gamma & \Omega(\gamma;y) \\
  \hline\hline
  \gamma_1 + k v_1 & 1\\
  -\gamma_1 + (k+1) v_1 & 1\\
  \hline
  v_1+k\gamma_{f} & y+y^{-1} \\
  -v_1+(k+1)\gamma_{f} & y+y^{-1} \\
  (k+1)\gamma_{f} & y^3 + 4y+y^{-1}\\
  \hline
\end{array}
\end{equation}
with $k\geq  0$ and
\begin{align}
v_1=\gamma_1+\gamma_2, && \gamma_f=\gamma_1+\gamma_2+\gamma_3.
\end{align}
This is almost the same as the BPS spectrum of the $\mathcal{N}=2^*$ theory from \cite{Longhi:2015ivt}, eq. (1.18): the D0 brane did not decouple, but became the flavour charge $\gamma_f$ of the adjoint hypermultiplet in the 4d theory $\gamma_f:=\gamma_1+\gamma_2+\gamma_3$, associated to the monodromy coordinate around the puncture. The difference is that in the $\mathcal{N}=2^*$ spectrum computed in \cite{Longhi:2015ivt} one has a single state $\gamma_1+\gamma_2+\gamma_f$, while from the deocupling we obtained an entire tower of states $\gamma_1+\gamma_2+k\gamma_f$. At the moment, the author does not have a geometrical or physical intuition for this peculiar degeneration, nor does he understands why the limiting BPS spectrum is almost, but not completely, identical to that of the $\mathcal{N}=2^*$ theory.

Another application of Theorem \ref{thm:ExSo} would be to study directly four-dimensional theories, where to the author's knowledge the existence of semiclassically exact solutions to the TBA equations \eqref{eq:TBA-conformal} has not previously been observed. The simplest example for which this would occur is the four-dimensional Argyres Douglas theory of type $(A_1,D_4)$, with BPS quiver shown in Figure \eqref{Fig:QuiverD4}. Theorem \ref{thm:ExSo} implies that for this theory the stability condition \eqref{eq:CC1defP1P1}, that was used in the case of $\mathbb{P}^1\times\mathbb{P}^1$, would provide a semiclassically exact solution in this case as well.

\subsubsection*{Higher del Pezzo and higher rank geometries}

We want to conclude by pointing out that, while all cases we discussed in this paper had only one compact divisor, corresponding to $SU(2)$ gauge theories and genus 1 mirror curves, this is not a requirement of Theorem \ref{thm:ExSo}. Indeed, it is important to generalize to geometries with more than one compact divisors, which would have $SU(N)$ gauge theory phases and higher genus mirror curves, and were related to cluster integrable systems of the Toda chain type in \cite{Bershtein:2018srt}. An obstacle to doing this is that the permutation symmetry alone does not seem to completely fix the BPS spectrum, as it happened in the rank-1 case. 

Another interesting generalization would be to higher local del Pezzos, that cannot be described as 5d uplift of UV complete four-dimensional gauge theories (for example, note  the absence of vertical arrows on the upper leftmost cases in the confluence diagram \ref{Fig:Sakai}), but are likely related to a five-dimensional uplift of the so-called $E_n$ Minahan-Nemeschansky theories \cite{Minahan:1996fg,Minahan:1996cj}, whose BPS spectrum displays complicated structures already in four dimensions, and was studied by spectral network methods in \cite{Hollands:2016kgm,Distler:2019eky,Hao:2019ryd}. While a tame collimation chamber might not exist for these cases, we saw that the wall-crossing invariant can be computed by finding marginally stable collimation regions.

\begin{appendix}{}
\section{Affine Weyl groups through cluster transformations}\label{App:WeylGroups}
This Appendix summarizes the realization in the cluster modular group of affine Weyl groups associated to local del Pezzos, as in \cite{Bershtein2017,Mizuno2020}. For a review on the relation between q-Painlev\'e equations and affine Weyl groups, see the review article \cite{kajiwara2017geometric} and the book \cite{Joshi2019Book}. In the following, we denote by $s_i$ a simple reflection along the $i$-th simple root of the lattice.
\subsection{Local $dP_5$ and $\mathcal{Q}(D_5^{(1)})$}\label{app:dP5}
The quiver for this case is shown in Figure \ref{Fig:QuiverdP5}. The flavor lattice is $\Gamma_f\simeq\mathcal{Q}(D_5^{(1)})$ with simple roots given by \cite{Mizuno2020}
\begin{align}
	\alpha_0=\gamma_2-\gamma_1, && \alpha_1=\gamma_6-\gamma_5, && \alpha_2=\gamma_1+\gamma_5, \nonumber\\
	\alpha_3=\gamma_3+\gamma_7, && \alpha_4=\gamma_4-\gamma_3, && \alpha_5=\gamma_8-\gamma_7,\label{eq:D51roots}
\end{align}
and null root
\begin{equation}
	\delta=\sum_{i=1}^8\gamma_i=\alpha_0+\alpha_1+2\alpha_2+2\alpha_3+\alpha_4+\alpha_5.
\end{equation}
The intersection pairing between the cycles corresponding to $\alpha_i$ in the underlying del Pezzo geometry is (minus) the Cartan matrix of $D_5^{(1)}$:
\begin{equation}\label{eq:CartanD5}
\left(\alpha_i|\alpha_j \right)^{dP_5}=-C_{ij}^{(D_5^{(1)})},\qquad C^{(D_5^{(1)})}=\left( \begin{array}{cccccc}
2 & 0 & -1 & 0 & 0 & 0 \\
0 & 2 & -1 & 0 & 0 & 0 \\
-1 & -1 & 2 & -1 & 0 & 0 \\
0 & 0 & -1 & 2 & -1 & -1 \\
0 & 0 & 0 & -1 & 2 & 0 \\
0 & 0 & 0 & -1 & 0 & 2
\end{array} \right).
\end{equation}
The generators of the extended affine Weyl group associated to simple reflections are realized through mutations and permutations as
\begin{align}
	s_0=(1,2), && s_1=(5,6), && s_2=(1,5)\mu_1\mu_5, \nonumber\\
	s_3=(3,7)\mu_3\mu_7, && s_4=(3,4), && s_5=(7,8),\label{eq:dP5refl}
\end{align}
while the Dynkin diagram $\mathbb{Z}_4$ automorphism 
\begin{equation}
\pi:\left(\alpha_0,\,\alpha_1,\,\alpha_2,\,\alpha_3,\,\alpha_4,\,\alpha_5 \right)\mapsto\left(\alpha_5,\,\alpha_4,\,\alpha_3,\,\alpha_2,\,\alpha_0,\,\alpha_1 \right)
\end{equation}
is realized as
\begin{equation}
	\pi=(1,3,5,7)(2,4,6,8).
\end{equation}
\subsection{Local $dP_4$ and $\mathcal{Q}(A_4^{(1)})$}
The quiver for this case is shown in Figure \ref{Fig:QuiverdP4}. The flavour lattice is $\Gamma_f\simeq\mathcal{Q}(A_4^{(1)})$, with simple roots given by
\begin{equation}
\alpha_0=\gamma_1+\gamma_3+\gamma_6, \quad \alpha_1=\gamma_2-\gamma_1, \quad \alpha_2=\gamma_1+\gamma_5, \nonumber
\end{equation}
\begin{equation}
\alpha_3=\gamma_3+\gamma_7, \qquad \alpha_4=\gamma_4-\gamma_3,
\end{equation}
and null root
\begin{equation}
\delta=\sum_{i=0}^{4}\alpha_i.
\end{equation}
The intersection pairing between the cycles corresponding to $\alpha_i$ in the underlying del Pezzo geometry is (minus) the Cartan matrix of $A_4^{(1)}$:
\begin{equation}\label{eq:CartanA4}
\left(\alpha_i|\alpha_j \right)^{dP_4}=-C_{ij}^{(A_4^{(1)})},\qquad C^{(A_4^{(1)})}=\left( \begin{array}{ccccc}
2 & -1 & 0 & 0 & -1  \\
-1 & 2 & -1 & 0 & 0  \\
0 & -1 & 2 & -1 & 0  \\
0 & 0 & -1 & 2 & -1  \\
-1 & 0 & 0 & -1 & 2 
\end{array} \right).
\end{equation}

\subsection{Local $dP_3$ and $\mathcal{Q}((A_2+A_1)^{(1)})$}
The quiver for this case is shown in Figure \eqref{Fig:QuiverdP3}. The flavour lattice is $\Gamma_f\simeq\mathcal{Q}((A_2+A_1)^{(1)})$, with simple roots given by
\begin{align}
\alpha_0=\gamma_3+\gamma_6 , && \alpha_1=\gamma_1+\gamma_4, && \alpha_2=\gamma_2+\gamma_5,\nonumber
\end{align}
\begin{align}
\beta_0=\gamma_2+\gamma_4+\gamma_6, && \beta_1=\gamma_1+\gamma_3+\gamma_5,
\end{align}
and null root
\begin{equation}
	\delta=\sum_{i=0}^{2}\alpha_i+\sum_{i=0}^1\beta_i.
\end{equation}
The extended affine Weyl group is generated by the reflections
\begin{align}\label{eq:dP3A2Refl}
s_0=(3,6)\mu_6\mu_3 && s_1=(1,4)\mu_4\mu_1, && s_2=(2,5)\mu_5\mu_2,
\end{align}
\begin{align}
r_0=(4,6)\mu_2\mu_4\mu_6\mu_2, && r_1=(3,5)\mu_1\mu_3\mu_5\mu_1,
\end{align}
and by the outer automorphisms generating ${\rm Dih}_6$
\begin{align}\label{eq:dP3Dih6}
\pi=(1,2,3,4,5,6), && \sigma=(1,4)(2,3)(5,6)\iota\,.
\end{align}
Of these, $\pi$ is an order-six Dynkin diagram automorphism that permutes simple roots, see Figure \ref{Fig:A1Dynkin2}.
The intersection pairing between the cycles corresponding to $\alpha_i,\beta_i$ in the underlying del Pezzo geometry is (minus) the Cartan matrix of $(A_2+A_1)^{(1)}$:
\begin{equation}\label{eq:CartanA2A1}
\left(\begin{array}{cc}\left(\alpha_i|\alpha_j \right)^{dP_3} & 0 \\ 0 & \left(\beta_i|\beta_j \right)^{dP_3}\end{array}\right)=-C_{ij}^{((A_2+A_1)^{(1)})},\qquad C^{((A_2+A_1)^{(1)})} =\left( \begin{array}{ccccc}
2 & -1 & -1 & 0 & 0  \\
-1 & 2 & -1 & 0 & 0  \\
-1 & -1 & 2 & 0 & 0  \\
0 & 0 & 0 & 2 & -2  \\
0 & 0 & 0 & -2 & 2 
\end{array} \right).
\end{equation}
We consider three affine translations that leave fixed the $A_1^{(1)}$ sublattice
\begin{align}\label{eq:T1T2T3}
	T_1=\pi^2s_2s_1, && T_2=\pi^2s_0s_2, && T_3=\pi^2s_1s_0, 
\end{align}
\begin{align}
	T_1(\vec{\alpha},\vec{\beta})=\left(\alpha_0+\delta,\alpha_1,\alpha_2-\delta,\vec{\beta}\right), && T_2(\vec{\alpha},\vec{\beta})=\left(\alpha_0-\delta,\alpha_1+\delta,\alpha_2,\vec{\beta}\right), && T_1T_2T_3=id, \nonumber
\end{align}
and an affine translation that leaves fixed the $A_2^{(1)}$ sublattice
\begin{align}\label{eq:T4}
	T_4=r_0\pi^3, && T_4(\vec{\alpha},\vec{\beta})=\left(\vec{\alpha},\beta_0+\delta,\beta_1-\delta\right).
\end{align}
\subsection{Local $dP_2$ and $\mathcal{Q}((A_1+A_1)^{(1)})$}
The quiver for this case is shown in Figure \eqref{Fig:QuiverdP2}. The flavour lattice is $\Gamma_f\simeq\mathcal{Q}((A_1+A_1)^{(1)})$, with simple roots given by
\begin{align}
\alpha_0=\gamma_2+\gamma_4+\gamma_5 , && \alpha_1=\gamma_1+\gamma_3,\nonumber
\end{align}
\begin{align}
\beta_0=2\gamma_1+\gamma_2+3\gamma_4-\gamma_5, && \beta_1=-\gamma_1+\gamma_3-2\gamma_4+2\gamma_5 ,
\end{align}
and null root
\begin{equation}
	\delta=\sum_{i=0}^{2}\alpha_i+\sum_{i=0}^1\beta_i.
\end{equation}
The extended affine Weyl group is generated by the reflections
\begin{align}\label{eq:dP3A2Refl}
s_0=(4,5)\mu_2\mu_5\mu_4\mu_2, && s_1=(1,3)\mu_1\mu_3,
\end{align}
The translation
\begin{align}
T=(2,1,4,5,3)(\mu_4),
\end{align}
and the involution
\begin{equation}
\sigma=\iota(1,3)(4,5)T.
\end{equation}
The intersection pairing between the cycles corresponding to $\alpha_i,\beta_i$ in the underlying del Pezzo geometry is (minus) the Cartan matrix of $(A_1+A_1)^{(1)}$ (with a nonstandard normalization of the simple roots for the second $A_1^{(1)}$ root lattice):
\begin{equation}\label{eq:CartanA1A1}
\left(\begin{array}{cc}\left(\alpha_i|\alpha_j \right)^{dP_2} & 0 \\ 0 & \left(\beta_i|\beta_j \right)^{dP_2}\end{array}\right)=-C_{ij}^{((A_1+A_1)^{(1)})},\qquad C^{((A_1+A_1)^{(1)})} =\left( \begin{array}{cccc}
2 & -2  & 0 & 0  \\
-2 & 2  & 0 & 0  \\
0 & 0 & 14 & -14  \\
0 & 0 & -14 & 14 
\end{array} \right).
\end{equation}

\subsection{Local $dP_1$ and $\mathcal{Q}(A_1^{(1)})$}
The quiver for this case is shown in Figure \eqref{Fig:QuiverdP1}. The flavour lattice is $\Gamma_f\simeq\mathcal{Q}(A_1^{(1)})$, with simple roots given by
\begin{align}
\alpha_0=\gamma_1-\gamma_2+2\gamma_3 , && \alpha_1=\gamma_2-\gamma_3+\gamma_4,\nonumber
\end{align}
and null root
\begin{equation}
	\delta=\alpha_0+\alpha_1.
\end{equation}
In this case the Cremona isometries lead only to a translation
\begin{equation}
T=(1,3,2,4)\mu_3
\end{equation}
and an involution
\begin{equation}
\sigma=\iota(1,4)(2,3).
\end{equation}
The intersection pairing between the cycles corresponding to $\alpha_i$ in the underlying del Pezzo geometry is (minus) the Cartan matrix of $A_1^{(1)}$ with an unusual normalization:
\begin{equation}\label{eq:CartanA1dP1}
\left(\alpha_i|\alpha_j \right)^{\mathbb{P}^1\times\mathbb{P}^1} =-C_{ij}^{(A_1^{(1)})},\qquad C^{(A_1^{(1)})} =\left( \begin{array}{cc}
8 & -8  \\
-8 & 8  \\
\end{array} \right).
\end{equation}
\subsection{Local $\mathbb{P}^1\times\mathbb{P}^1$ and $\mathcal{Q}(A_1^{(1)})$}
The quiver for this case is shown in Figure \eqref{Fig:QuiverP1P1}. The flavour lattice is $\Gamma_f\simeq\mathcal{Q}(A_1^{(1)})$, with simple roots given by
\begin{align}
\alpha_0=\gamma_1+\gamma_3 , && \alpha_1=\gamma_2+\gamma_4,\nonumber
\end{align}
and null root
\begin{equation}
	\delta=\alpha_0+\alpha_1.
\end{equation}
The extended affine Weyl group is generated by the reflections
\begin{align}\label{eq:dP3A2Refl}
s_0=(1,3)\mu_1\mu_3, && s_1=(2,4)\mu_2\mu_4
\end{align}
and the involution
\begin{equation}
\pi=(1,2,3,4).
\end{equation}
The intersection pairing between the cycles corresponding to $\alpha_i$ in the underlying del Pezzo geometry is (minus) the Cartan matrix of $A_1^{(1)}$ :
\begin{equation}\label{eq:CartanA1F0}
\left(\alpha_i|\alpha_j \right)^{\mathbb{P}^1\times\mathbb{P}^1} =-C_{ij}^{(A_1^{(1)'})},\qquad C^{(A_1^{(1)'})} =\left( \begin{array}{cc}
2 & -2  \\
-2 & 2  \\
\end{array} \right).
\end{equation}

\section{Discrete equations for 4d Super Yang-Mills, and the 2-Kr\"onecker quiver}\label{App:4dSYM}

I remarked in the main text that q-Painlev\'e equations can be regarded as Y-systems of the TBA equations \eqref{eq:TBA-conformal}. The usual notion of Y-system, however, arises by shifting the phase of $\epsilon$, while in q-Painlev\'e we have shifts of the moduli. This seems to be the correct generalization of a Y-system to the case of an infinite chamber, and is not limited to\footnote{I am grateful to K. Ito for many interesting discussions about Y-systems and the ODE/IM correspondence, that led me to better appreciate this point.} five-dimensional theories. To show this, I will now show how a similar perspective can be adopted in the case of weakly coupled 4d Super Yang-Mills, corresponding to the infinite chamber of the 2-Kr\"onecker quiver. 
\begin{figure}[h]
\begin{center}
\includegraphics[width=.35\textwidth]{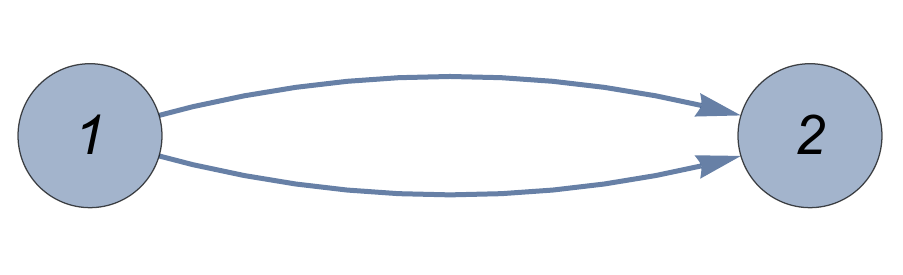}
\end{center}
\caption{BPS quiver of four-dimensional pure $SU(2)$ SYM}
\label{Fig:QuiverN2star}
\end{figure}
The moduli space of stability conditions is divided into two chambers: a finite (strongly coupled) chamber with $\arg Z_1<\arg Z_2$, and an infinite (weakly coupled chamber with $\arg Z_1>\arg Z_2$. The spectrum in the weakly coupled chamber consists of a vector multiplet with charge $\gamma_1+\gamma_2$, and infinite towers of dyons $\pm\gamma_1+k(\gamma_1+\gamma_2)$, $k\in\mathbb{Z}$. The central charges are periods of the Seiberg-Witten differential $\lambda_{SW}:=y\dd x$ on the Seiberg-Witten curve
\begin{equation}
	\Sigma:\,y^2=Q_0(x),\qquad Q_0(x)x^2=\Lambda^2\left(x+\frac{1}{x} \right)+u.
\end{equation}
Let $a:=\lim_{\Lambda\rightarrow 0}Z_{\gamma_1+\gamma_2}$, and let $\eta\sim\lim_{\Lambda\rightarrow 0}Z_1$, where $\sim$ includes a renormalization factor because $\eta$ is singular in the limit. We can view these as Fenchel-Nielsen coordinates, since the limit $\Lambda\rightarrow 0$ corresponds to a degeneration of the Seiberg-Witten curve. The requirement of a renormalization factor leaves ambiguities in the definition of $\eta$, which however do not affect our discussion. See \cite{Coman:2020qgf} for a discussion on how to fix such ambiguities, but in our present case this albiguities correspond to the choice of different solutions in our difference equations.
\begin{figure}[h]
\begin{center}
\includegraphics[width=.7\textwidth]{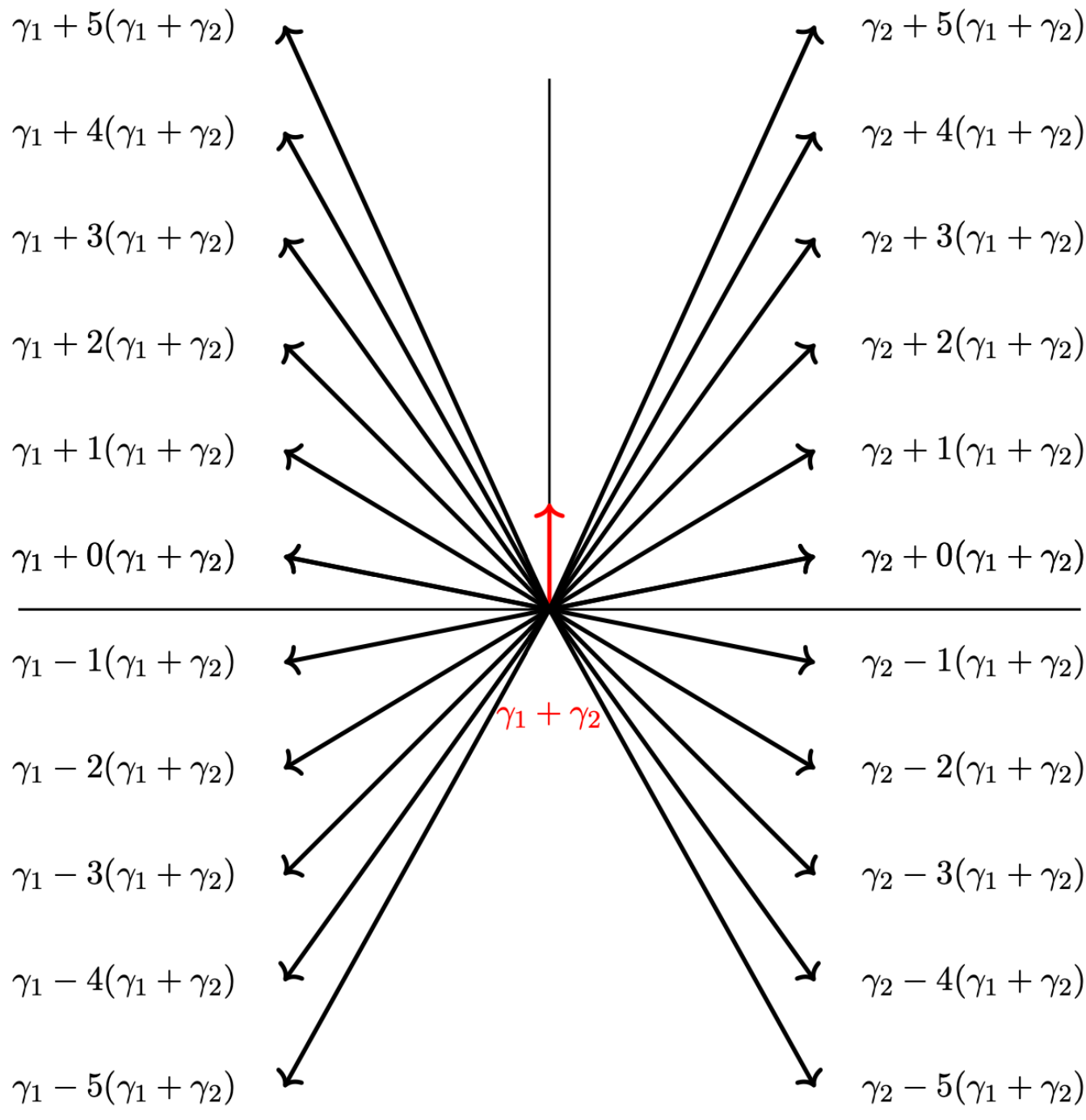}
\end{center}
\caption{Weakly-coupled spectrum of four-dimensional pure $SU(2)$ SYM}
\label{Fig:QuiverN2star}
\end{figure}
Using the same arguments of Section \ref{sec:dP5}, observe that if we send $Z_1\mapsto Z_1+(Z_1+Z_2)$, $Z_2\mapsto Z_2-(Z_1+Z_2)$, which we can interpret as the discrete time evolution $T=(1,2)\mu_1$, which acts on our coordinates as
\begin{equation}
	T(a,\eta)=(a,\eta+a).
\end{equation}
It leads to the find the following equations for $\CX_1,\,\CX_2$:
\begin{equation}
	\CX_1(a,\eta+a)\CX_1(a,\eta-a)=\left(1+\CX_1(a,\eta)^{-1} \right)^{-2},\nonumber
\end{equation}
\begin{equation}
	\CX_2(a,\eta+a)\CX_2(a,\eta-a)=\left(1+\CX_2(a,\eta) \right)^2. 
\end{equation}
Similar equations appeared in \cite{Cecotti:2014zga}, eq. 2.18, where a ``would be Y-system" associated to this chamber was written as a formal recurrence relation $Y_{n+1}Y_{n-1}=(1+Y_n)^2$ (in fact this equation can already be found in Appendix A of \cite{Gaiotto:2008cd}, where also an explicit solution to the recursion is provided), which would correspond to the second of the above equations. However, the discrete variable $n$ therein had the meaning of a shift in the phase of $\hbar$. For infinite chambers, such a shift is not constant, so it is not clear how the recursion relation can be turned into a difference equation involving $\hbar$. Instead, it is easy to check that these equations are solved by
\begin{align}
	\CX_1(a,\eta)=-\left(\frac{\sin2\pi a}{\sin2\pi\eta} \right)^2, && \CX_2(a,\eta)=-\left(\frac{\sin2\pi(a+\eta)}{\sin2\pi a} \right),
\end{align}
which matches exactly with the change from Fock-Goncharov coordinates $\CX_i$ to Fenchel-Nielsen coordinates $a,\,\eta$ from \cite{Coman:2020qgf}. This should also be related to equation (5.44) in \cite{Grassi:2019coc}, representing the transformation between Borel resummed quantum periods and quantum periods from instanton counting.
\newpage
\section*{Statements and Declarations}

\textbf{Data availability statement:} Data sharing not applicable to this article as no datasets were generated or analysed during the current study.

\noindent\textbf{Competing interests:} The author has no relevant financial or non-financial interests to disclose.

\noindent\textbf{Funding:} The early stages of this work were completed during the author's residence at the Isaac Newton Institute during the Fall 2022 semester and he thanks the M\o ller institute and organizers of the program "Applicable resurgent asymptotics: towards a universal theory" for hospitality during his stay, supported by EPSRC grant no EP/R014604/1.

\end{appendix}

\bibliographystyle{alph}
\bibliography{biblio.bib}

\end{document}